\documentclass[a4paper,11pt]{article}

\usepackage{fullpage}
\usepackage{amsmath,amsthm,amssymb}
\usepackage{graphicx}
\usepackage{caption,subcaption}

\usepackage{algorithm}
\usepackage{algpseudocode}

\algtext*{EndWhile} 
\algtext*{EndIf} 
\algtext*{EndFor} 

\newtheorem{lemma}{Lemma}
\newtheorem{theorem}[lemma]{Theorem}
\newtheorem{corollary}[lemma]{Corollary}
\newtheorem{proposition}[lemma]{Proposition}

\newcommand{\eps}{\varepsilon}

\usepackage{color}

\usepackage{hyperref}

\let\geq\geqslant
\let\leq\leqslant

\DeclareMathOperator{\Top}{Top}
\DeclareMathOperator{\blocked}{Blocked}

\title{A Faster Algorithm for Computing Motorcycle Graphs}

\author{Antoine Vigneron
\thanks{
King Abdullah University of Science and
Technology (KAUST),
Geometric Modeling and Scientific Visualization
Center, Thuwal 23955-6900, Saudi Arabia.
{\tt \{antoine.vigneron, lie.yan\}@kaust.edu.sa}}
\and 
Lie Yan\footnotemark[1]
}

\begin{document}
\maketitle

\begin{abstract}
We present a new algorithm for computing motorcycle graphs
that runs in $O(n^{4/3+\eps})$ time for any $\eps>0$,
improving on all previously known algorithms. 
The main application of this result is to computing
the straight skeleton of a polygon. It allows us 
to compute the straight skeleton of a non-degenerate polygon
with $h$ holes in $O(n \sqrt{h+1} \log^2 n+n^{4/3+\eps})$ expected time.
If all input coordinates are $O(\log n)$-bit 
rational numbers, we can compute the straight skeleton of a (possibly degenerate)
polygon with $h$ holes  in $O(n \sqrt{h+1}\log^3 n)$ expected time. 

In particular, it means that we can compute the straight skeleton
of a simple polygon in $O(n\log^3n)$ expected time if all input
coordinates are $O(\log n)$-bit rationals, while all previously known
algorithms have worst-case running time $\omega(n^{3/2})$.
\end{abstract}

\section{Introduction}

The straight skeleton of a polygon $P$ is a straight line graph
embedded in $P$, formed by the traces
of the vertices of $P$ when it is shrunk, each edge moving
at the same speed and remaining parallel to its original position.
(See \figurename~\ref{fig:mgdef}.)
\begin{figure}
\centering
        \begin{subfigure}[b]{.3\textwidth}
        \centering
        \includegraphics[width=\textwidth]{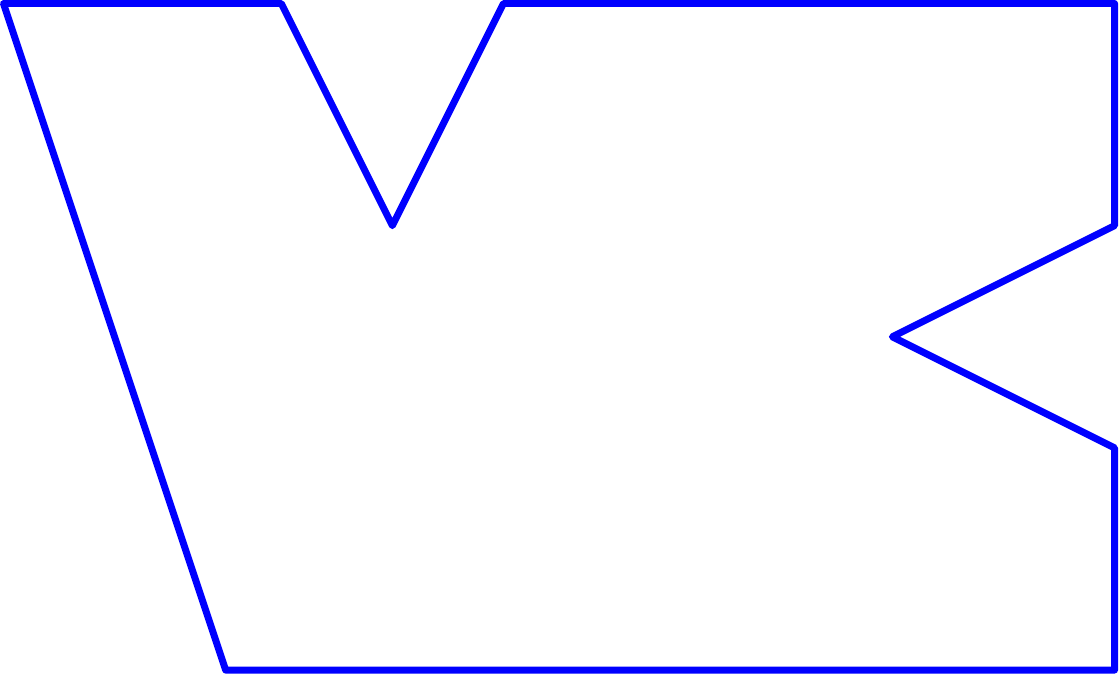}
        \caption{Input polygon}
        \end{subfigure}
        \quad
        \centering
        \begin{subfigure}[b]{.3\textwidth}
        \includegraphics[width=\textwidth]{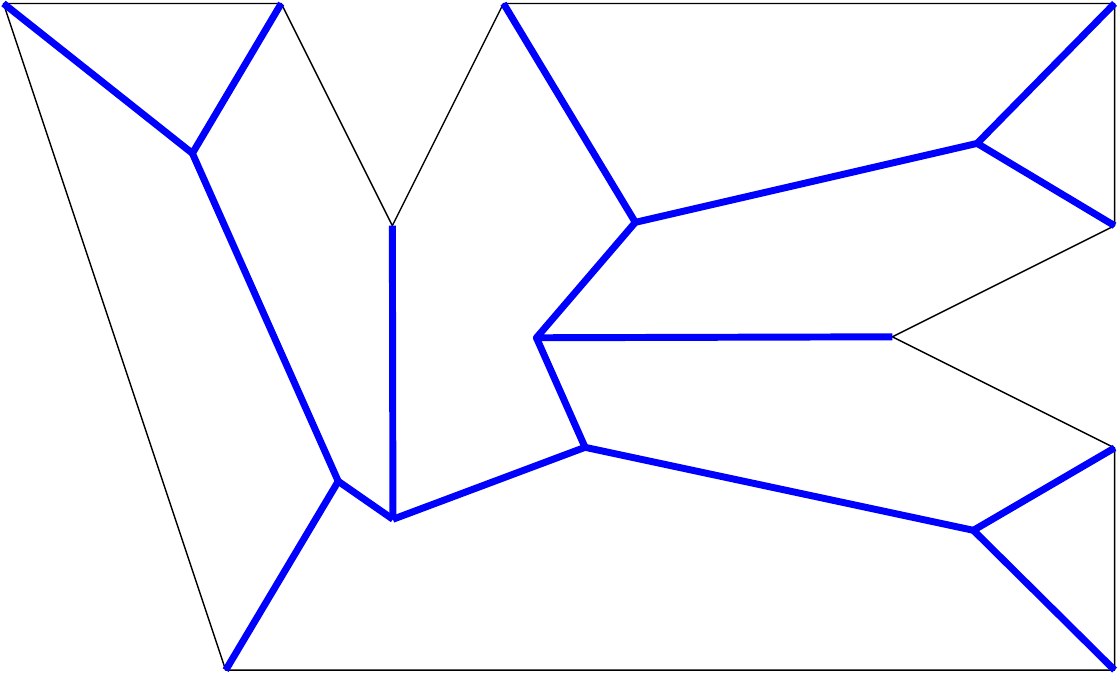}
        \caption{Straight skeleton}
        \end{subfigure}
        \quad
        \centering
        \begin{subfigure}[b]{.3\textwidth}
        \includegraphics[width=\textwidth]{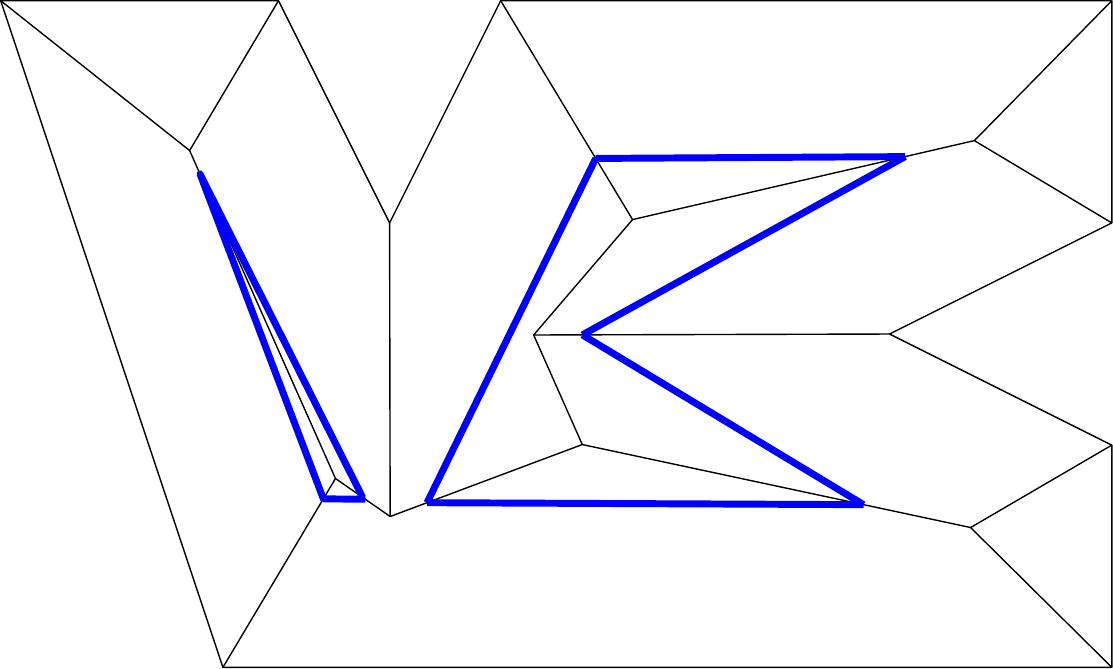}
        \caption{Offset polygon}
        \end{subfigure}
\caption{The straight skeleton of a polygon, and a corresponding offset
polygon.\label{fig:mgdef}}
\end{figure}
It has been known since at least the 19th century;
for instance, figures representing the straight skeleton 
can be found in the book by von Peschka~\cite{peschka-1877}.
Aichholzer et al.~\cite{aichholzer1995novel} gave the first efficient
algorithms for computing the straight skeleton, 
and presented it as an alternative to the medial axis having only straight-line edges. 
The straight skeleton has found numerous applications in computer
science, for instance to city 
model reconstruction~\cite{ld03}, architectural modeling~\cite{Kelly11},
polyhedral surface reconstruction~\cite{bgls03,Fel98,Oli96},
biomedical image processing~\cite{Clo00}.
It also has a direct application to CAD,
as it allows to compute  offset polygons~\cite{eppstein1999raising}.
The straight skeleton has become a standard tool in geometric computing, 
and thus fast and robust software has been developed to compute 
it~\cite{Cacciola,HuberH11,PalfraderHH12}.

The complexity of straight skeleton computation, however, is still
very much open. The previously best known algorithms were the
$O(n^{17/11+\eps})$-time algorithm by Eppstein and 
Erickson~\cite{eppstein1999raising}, and the 
$O(n^{3/2}\log^2 n)$-time randomized algorithm by Cheng and 
Vigneron~\cite{ChengV07}. The only known lower bound is
$\Omega(n \log n)$, by a reduction from sorting~\cite{jeffwebpage}.
In this paper, we give new subquadratic algorithms for computing
straight skeletons. In particular, if all input coordinates 
are $O(\log n)$-bit rational numbers, 
we give an $O(n\sqrt{h+1}\log^3 n)$-time randomized
algorithm for computing the straight skeleton of a polygon
with $h$ holes. It is the first near-linear time algorithm
for computing the straight skeleton of a simple polygon.

Eppstein and Erickson~\cite{eppstein1999raising} introduced
{\em motorcycle graphs} so as to model the main difficulty of straight
skeleton computation. We are given a set of $n$ motorcycles,
each motorcycle having a starting point and a velocity. Each motorcycle 
moves at constant velocity until it reaches the track left by another 
motorcycle, in which case it crashes. The resulting graph is called a 
motorcycle graph. (See \figurename~\ref{fig:mgss}a.)
\begin{figure}
\centering
        \begin{subfigure}[b]{.4\textwidth}
        \centering
        \includegraphics[width=\textwidth]{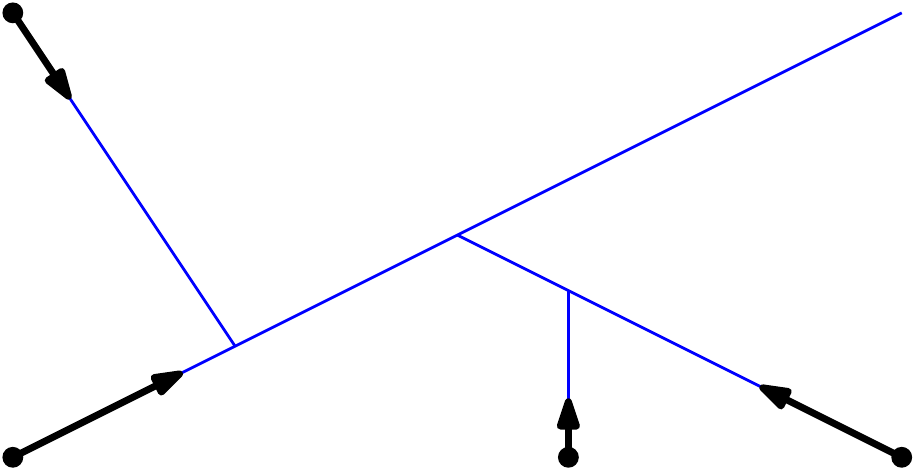}
        \caption{A motorcycle graph}
        \end{subfigure}
        \qquad \qquad
        \centering
        \begin{subfigure}[b]{.4\textwidth}
        \includegraphics[width=\textwidth]{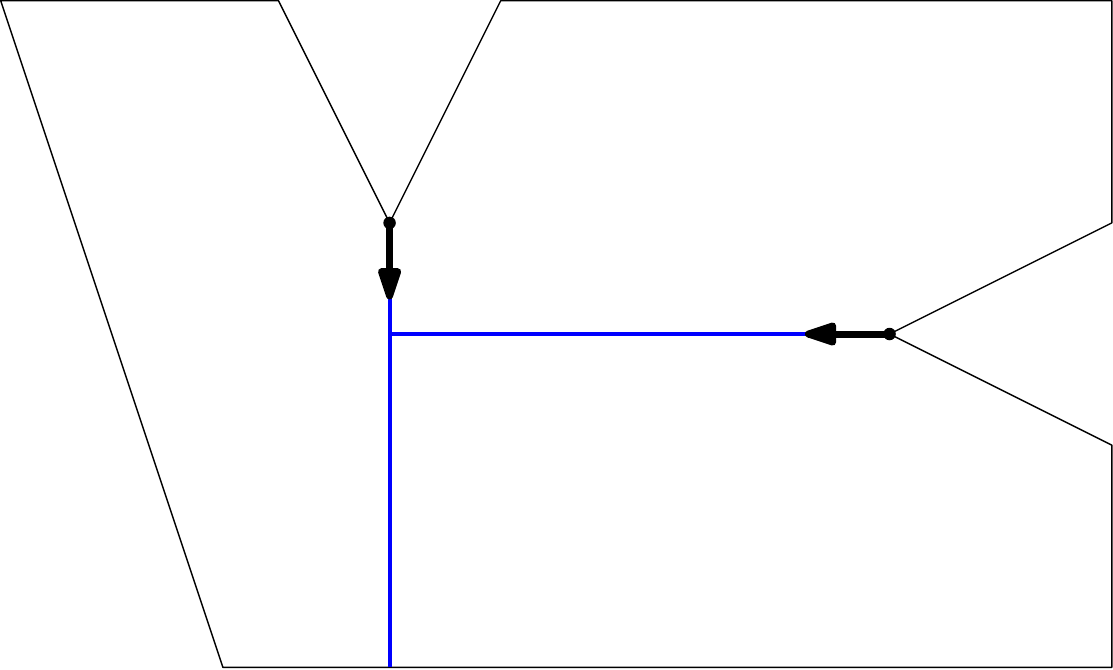}
        \caption{Induced motorcycle graph}
        \end{subfigure}
\caption{The motorcycle graph of a set of four motorcycles (a), and
the motorcycle graph induced by the polygon from \figurename~\ref{fig:mgdef}.
\label{fig:mgss}}
\end{figure}
The motorcycle graph is a special case of the straight 
skeleton, where each motorcycle is modeled by a small and thin triangle.
Conversely, a polygon induces a motorcycle graph, where each
motorcycle starts at a reflex vertex and moves with the same
velocity as this vertex moves during the shrinking process.
(See \figurename~\ref{fig:mgss}b.)
Cheng and Vigneron~\cite{ChengV07} showed that computing
the straight skeleton of a non-degenerate polygon 
reduces to computing this induced motorcycle graph, and a lower 
envelope computation; Huber and Held extended this proof to 
degenerate cases~\cite{HuberH11}. The lower envelope computation
can be done in $O(n\sqrt{h+1}\log^2 n)$ expected time if $P$
has $h$ holes. 

Previously, the bottleneck of straight skeleton computation
was the induced motorcycle graph computation. 
This is our main motivation for designing a faster motorcycle
graph algorithm. In this paper, we give an algorithm for 
computing a motorcycle graph that runs in $O(n^{4/3+\eps})$ time,
for any $\eps>0$, improving on all previously known algorithms.
Here is a brief description of our algorithm. For each
motorcycle, we maintain a tentative
track, which may be longer than its actual track
in the motorcycle graph. We also maintain a set of target points, 
which contains the endpoints of the
tentative tracks that have been created earlier for this motorcycle,
and that it  has not reached yet.
Initially, the tentative tracks are empty, and
then we try to extend them one by one,
all the way to the destination point. If two tentative tracks
cross, we retract them, by roughly halving the number of
possible crossing points on each of them. After performing this
halving, the tentative tracks do not intersect, and we can
safely move the motorcycle that reaches the end of its 
tentative track first. Then we try to extend the tentative
track of this motorcycle to its next target point, and repeat the 
process. An example is given in Appendix~\ref{sec:example}.

Apart from obtaining better time bounds for straight skeleton computation,
there are at least two other reasons for studying motorcycle
graphs. First, Huber and Held~\cite{HuberH11}
used the idea of computing the straight skeleton
from its induced motorcycle graph to design and
implement a practical straight skeleton algorithm. So it is important,
even in practice,
to get a better understanding of motorcycle graph computation.
Another motivation for studying motorcycle graphs is a direct application 
to  computer graphics, for quad mesh partitioning~\cite{EppsteinGKT08}.

Some of our results make no particular assumptions on the input,
but we also present a few results where we assume that the input coordinates 
are $O(\log n)$-bit rational numbers. We believe that this assumption is
sufficient for most applications. For instance, in the applications
mentioned above, it is hard to imagine that the input polygons would have
features smaller than  1nm, and size larger than 1000km,
so 64-bit integers should be more than sufficient.

\subsection{Summary of our results and comparison with previous work}
\label{sec:ourresults}

The main novelty in this paper is our algorithm for computing 
motorcycle graphs (Section~\ref{sec:algorithm}). 
This algorithm is essentially different from
the two previous algorithms~\cite{ChengV07,eppstein1999raising}
that both simulate the construction in chronological order.
Our algorithm, on the other hand, does not construct the
motorcycle tracks in chronological order: It may move some motorcycle
to its position at time $t$, and then later during the execution
of the algorithm, move another motorcycle
to its position at an earlier time $t'<t$. (We give one such example
in Appendix~\ref{sec:example}.) This answers an open question
by Eppstein and Erickson~\cite[end of Section 5]{eppstein1999raising},
who asked whether the running time can be improved by relaxing
the chronology of the events.

Our algorithm uses two auxiliary
data structures, one for ray shooting, and another for halving queries.
Given a query segment on the supporting line of a motorcycle, 
a halving query returns a splitting point on this segment such
that there are roughly the same number of intersections with other
supporting lines on both sides.
(See Section~\ref{sec:notation}.) The implementation of these data 
structures in different settings lead to different time
bounds. 

For all our results, we use the standard real-RAM model~\cite{PreparataShamos}, 
that allows to perform arithmetic operations exactly on arbitrary
real numbers. But for some of our results, we make
the assumption that all input coordinates are $O(\log n)$-bit rational 
numbers. It has two advantages: It yields better time bounds, and
allows us to handle the straight skeleton of degenerate
polygons. This improvement comes from the fact that, for bounded precision
input, two distinct crossing points between the supporting lines
of two pairs of motorcycles are at distance $2^{-O(\log n)}$ from each
other. It allows us to use a simpler halving scheme: Instead of halving
a segment according to the number of intersection points, we use the midpoint 
according to the Euclidean distance.
(See Section~\ref{sec:bpmg} and~\ref{sec:prelss}.)

\paragraph{Arbitrary precision input.} For our first set of results,
the input coordinates are arbitrary real numbers, on which we can perform 
exact arithmetic operations. In this case,
our new algorithm computes a motorcycle graph in $O(n^{4/3+\eps})$ time
(Theorem~\ref{th:general}).
This improves on the two subquadratic algorithms that
were known before: the $O(n^{17/11+\eps})$-time algorithm 
by Eppstein and Erickson~\cite{eppstein1999raising}, which
was first published in 1998, 
and the $O(n^{3/2}\log n)$-time algorithm by Cheng and 
Vigneron~\cite{ChengV07}, which first appeared in 2002.

We also give, in Section~\ref{sec:coriented}, 
an $O(C n \log^2(n)\min(C,\log n))$-time algorithm
for the case of $C$-oriented motorcycles, where the velocities
take only $C$ different directions. 
This improves on 
the algorithm by Eppstein and Erickson~\cite{eppstein1999raising},
which runs in $O(n^{4/3+\eps})$ time when $C=O(1)$.

Our last result with arbitrary precision input is an 
$O(n^{4/3+\eps}+n \sqrt{h+1}\log^2 n)$ expected time algorithm for computing
the straight skeleton of a polygon with $n$ vertices and $h$ holes.
(This result does not hold for a degenerate polygon where two reflex vertices
may collide during the shrinking process, as in \figurename~\ref{fig:degenerate}.)
\begin{figure}
\centering
        \begin{subfigure}[b]{.3\textwidth}
        \centering
        \includegraphics[width=\textwidth]{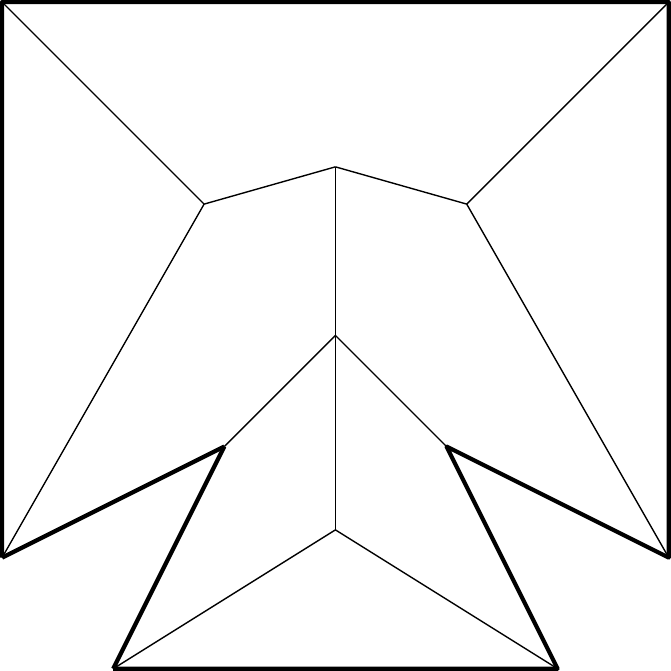}
        \end{subfigure}
        \quad
        \begin{subfigure}[b]{.3\textwidth}
        \centering
        \includegraphics[width=\textwidth]{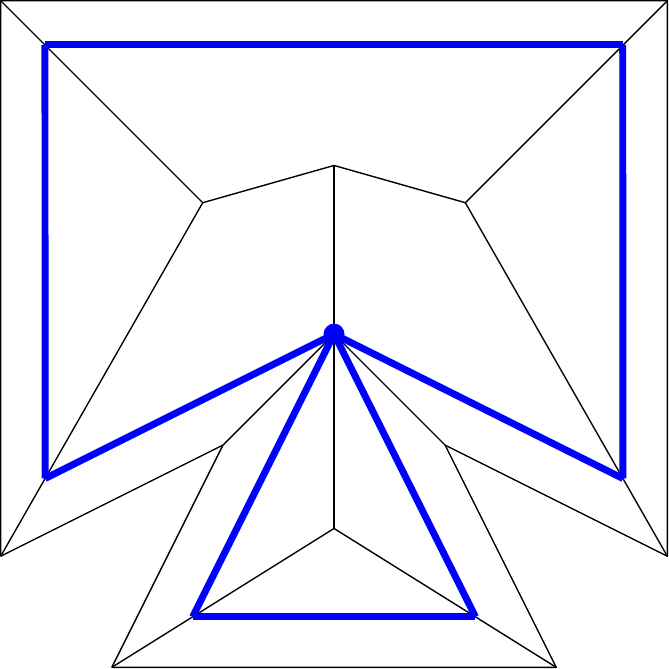}
        \end{subfigure}
        \quad
        \begin{subfigure}[b]{.3\textwidth}
        \centering
        \includegraphics[width=\textwidth]{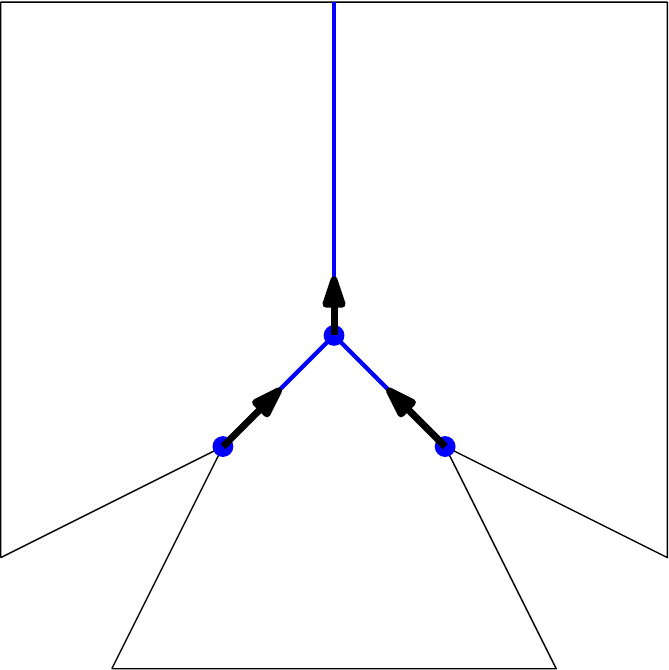}
        \end{subfigure}
\caption{A degenerate polygon and its straight skeleton (left).
Two reflex vertices collide during the shrinking process, and a new
reflex vertex appears (middle). The induced motorcycle graph,
where a new motorcycle appears when the two other crash (right).
\label{fig:degenerate}}
\end{figure}
It improves on the algorithm
by Cheng and Vigneron~\cite{ChengV07} which runs in 
$O(n^{3/2}\log(n)+n \sqrt{h+1}\log^2 n)$ expected time. It also improves
on the $O(n^{17/11+\eps})$ time bound of the algorithm by 
Eppstein and Erickson~\cite{eppstein1999raising}, but their algorithm
is deterministic and applies to degenerate cases.

\paragraph{Bounded precision input.} The following results
hold when all input coordinates
are $O(\log n)$-bit rational numbers. There has been recent
interest in studying computational geometry problems under
a bounded precision model (the word RAM), for instance
the computation of Delaunay triangulations, convex hulls, polygon
triangulation and line segment intersections~\cite{BuchinM11,ChanP09}.

We first show in Section~\ref{sec:bpmg} that 
a motorcycle graph can be computed in $O(n \log^3 n)$ time if
the motorcycles move within a simple polygon, starting
from its boundary. The only other non-trivial cases where we
know how to compute a motorcycle graph in near-linear time seem to be
the case where all velocities have positive $x$-coordinate,
which can be solved in $O(n\log n)$ time by plane sweep, 
the case of a constant number of different velocity
vectors~\cite{eppstein1999raising}, or a constant number
of directions (Section~\ref{sec:coriented}). 

Then in Section~\ref{sec:bpss}, we show that the straight
skeleton of a polygon with $n$ vertices and $h$ holes
can be computed in $O(n\sqrt{h+1}\log^3 n)$ expected time. This
result still holds in degenerate cases. So with bounded-precision
input, and if we allow randomization, it improves on the 
$O(n^{17/11+\eps})$-time algorithm by Eppstein and 
Erickson~\cite{eppstein1999raising}.
When $h=o(n/\log^2 n)$, it also improves on the $O(n^{3/2}\log^2 n)$-time 
algorithm by Cheng and Vigneron~\cite{ChengV07}, 
which cannot handle all degenerate cases. 

In particular, our algorithm runs in expected $O(n\log^3 n)$ time 
when $h=0$, so it is the first
near-linear time algorithm for computing the straight skeleton
of a simple polygon. The previously
best known algorithms run in $\omega(n^{3/2})$ time
in the worst case~\cite{ChengV07,eppstein1999raising}.

\subsection{Notation and preliminaries}\label{sec:notation}
For any two points $p,q$, we denote by $\overline{pq}$ the
line segment between $p$ and $q$. Unless specified otherwise,
$\overline{pq}$ is a closed segment. The {\em relative interior}
of $\overline{pq}$ is the open segment $\overline{pq}\setminus\{p,q\}$.
We say that two segments {\em cross} if their relative interiors
intersect.

The motorcycles are numbered from $1$ to $n$. Each
motorcycle $i$ has a {\em starting point} $s_i$, moves with constant
velocity $\vec v_i$, and has a {\em destination point} $d_i$ that lies
in the ray $(s_i,\vec v_i)$. 
(See \figurename~\ref{fig:notation}a.)
\begin{figure}
\centering
        \begin{subfigure}[b]{.55\textwidth}
        \centering
        \includegraphics[width=\textwidth]{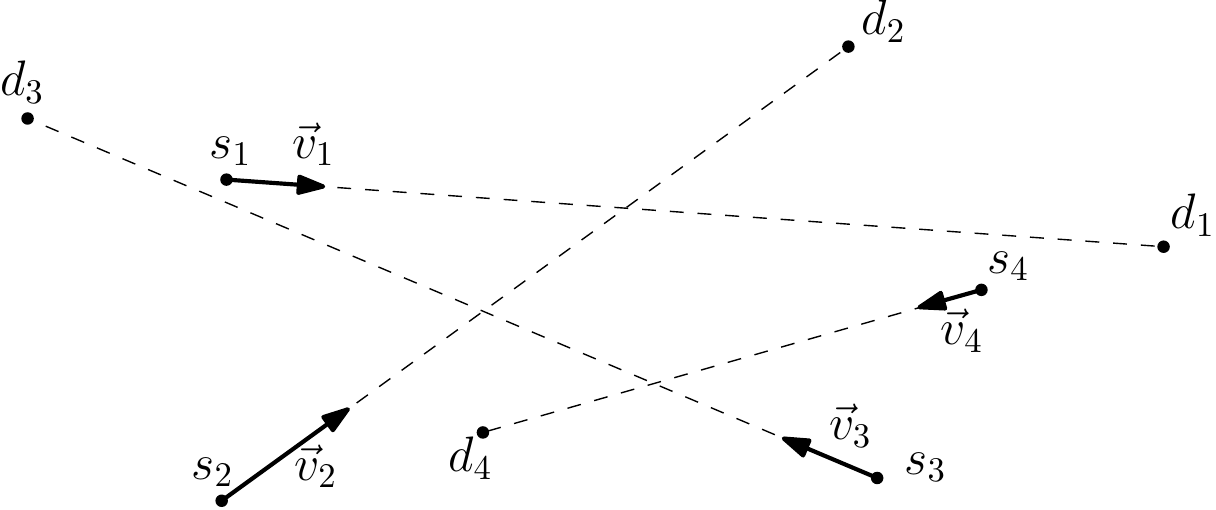}
        \caption{Input}
        \end{subfigure}
        \qquad
        \centering
        \begin{subfigure}[b]{.35\textwidth}
        \includegraphics[width=\textwidth]{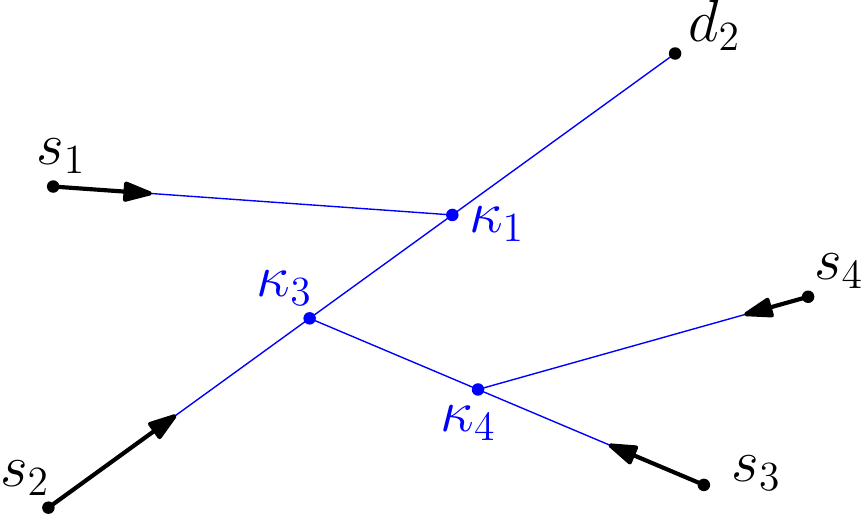}
        \caption{Motorcycle graph}
        \end{subfigure}
\caption{The input to the motorcycle graph problem (a),
and the resulting motorcycle graph (b).
\label{fig:notation}}
\end{figure}
When $p \in \overline{s_id_i}$, we 
denote by $\tau(i,p)$ the time when motorcycle $i$ reaches $p$, 
so $p=s_i+\tau(i,p)\vec{v_i}$.
The {\em supporting line} $\ell_i$ of motorcycle $i$ is the
line through $s_i$ with direction $\vec v_i$. 

Each motorcycle $i$ starts at $s_i$ at time 0,
and moves at velocity $\vec v_i$ until it meets the track 
left by another motorcycle and crashes, or it reaches $d_i$
and stops. So motorcycle $i$ crashes if it reaches a point
$p$ such that $\tau(i,p)\geq \tau(j,p)$, for some motorcycle $j$
that has not crashed or stopped earlier than $\tau(j,p)$.
If motorcycle $i$ crashes, we denote by $\kappa_i$ the point where
it crashes, called the {\em crashing point}. 
(See \figurename~\ref{fig:notation}b.)
Otherwise,
$i$ reaches $d_i$, and we set $\kappa_i=d_i$. The {\em trajectory}
of $i$ is the segment $\overline{s_i\kappa_i}$; in other words it is the
track of $i$ in the motorcycle graph.  

In the original motorcycle graph problem, the destination point
$d_i$ is at infinity in direction $\vec v_i$. We can handle this case by computing
a bounding box that includes all the vertices of the arrangement
of the supporting lines $\ell_i$, $i=1,\dots,n$,  and choosing
as destination points the intersections of the rays $(s_i,\vec v_i)$
with the bounding box. The bounding box can be computed 
in $O(n \log n)$ time as any extreme vertex in the arrangement is the 
intersection of two lines with consecutive slopes.

Unless specified otherwise, we make the following general position assumptions.
No two motorcycles share the same supporting line, or have
parallel supporting lines.
No three supporting lines are concurrent. No point $s_i,d_i$ lies on $\ell_j$
if $j \neq i$. No two motorcycles reach the same point at the same time.
(We make these assumptions so as to simplify the description of
the algorithm and the proofs, but our results still hold in degenerate
cases.)

The {\em crossing point} $\chi_{ij}$ is the intersection between $\ell_i$ 
and $\ell_j$, and thus $\chi_{ij}=\chi_{ji}=\ell_i \cap \ell_j$. 
The {\em size} $|pq|$ of a segment $\overline{pq}$
is the number of crossing points $\chi_{ij}$ that 
lie in $\overline{pq}$. (See \figurename~\ref{fig:size}.)
\begin{figure}
\begin{center}
\includegraphics[scale=.9]{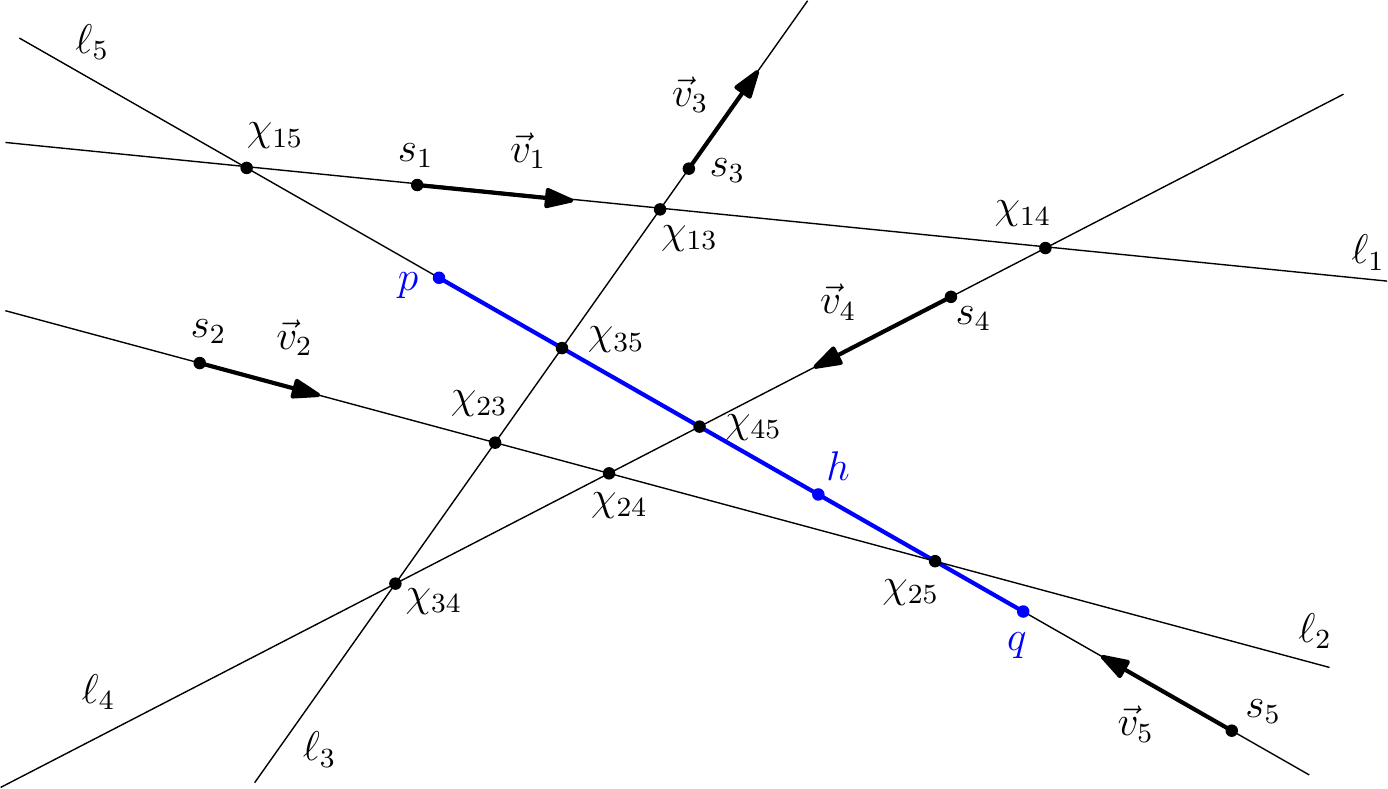}
\end{center}
\caption{The size of $\overline{pq}$ is $|pq|=3$.
Point $h$ is a possible result $h(p,q)$ of a halving query
$(5,p,q)$ with $\rho=1/2$. \label{fig:size}}
\end{figure}
We will  need a data structure 
to answer {\em halving queries}: Given a query $(i,p,q)$
where $p,q$ are points on the supporting line $\ell_i$,
find a point $h=h(p,q) \in \overline{pq}$ such that 
$|ph| \leq \lceil \rho|pq|\rceil$ and $|hq| \leq \lceil \rho|pq| \rceil$,
for a constant $\rho<1$. In addition, we require that $h$ is
not a crossing point, and that both $|ph|$ and $|hq|$ are strictly 
smaller than $|pq|$ if $|pq|\geq 2$.

\section{Algorithm for Computing Motorcycle Graphs}\label{sec:main}

In this section, we present our algorithm for computing motorcycle
graphs, as well as its proof of correctness and analysis.
An example of the execution of this algorithm on a set of
4 motorcycles is given in Appendix~\ref{sec:example}. 

\subsection{Algorithm description}\label{sec:algorithm}

Our algorithm maintains, for each motorcycle $i$, a {\em confirmed track} 
$\overline{s_ic_i}$, and a {\em tentative track} $\overline{s_it_i}$,
such that $c_i \in \overline{s_id_i}$ and $t_i \in \overline{c_id_i}$.
So the tentative track is at least as long as the confirmed track.
As we will show in the next section, the confirmed track is a subset of 
the trajectory, so we have   $c_i \in \overline{s_i\kappa_i}$ at 
any time during the execution of the algorithm. The tentative 
track, however, may go beyond $\kappa_i$. 
(See Appendix~\ref{sec:example}.)

Our algorithm will ensure that
no two tentative tracks cross.  We keep all the tentative tracks in a
ray shooting data structure, so that we can enforce this invariant by checking
for intersection each time we try to extend a tentative track. This data structure
returns the first tentative track hit by a query ray $(p,\vec v)$, if any.
We also build a data
structure to answer halving queries, which will be used  to shorten tentative
tracks and keep them disjoint.

Our algorithm builds the motorcycle graph by extending the confirmed tracks
until they form the whole motorcycle graph. 
We may also update the tentative track of a motorcycle when we extend its
confirmed track. 

A set of {\em target points} is associated with each motorcycle $i$. 
In particular, we maintain in a stack $S_i$ 
the set of target points that lie beyond the confirmed track
of motorcycle $i$, thus
$S_i \subset \overline{c_id_i}\setminus \{c_i\}$.
In other words, $S_i$ records the target points that motorcycle $i$
has not reached yet. 
(See \figurename~\ref{fig:tracks}.)
\begin{figure}
\begin{center}
\includegraphics[scale=.75]{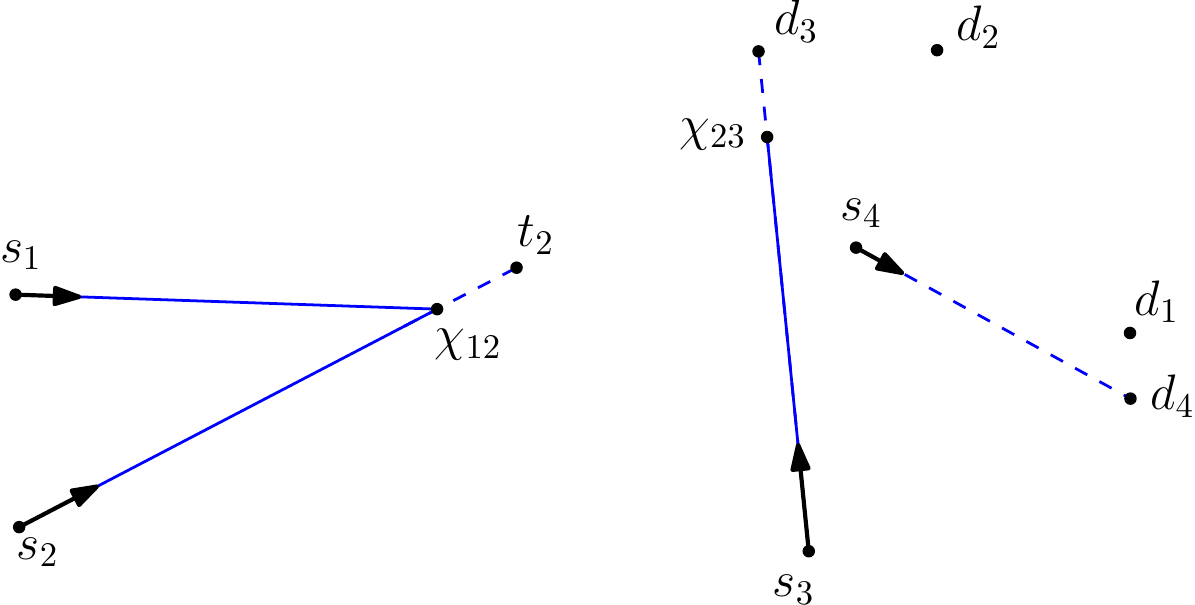}
\end{center}
\caption{This is the same example as Appendix~\ref{sec:example},
\figurename~(m).
The confirmed tracks are solid, and the tentative
tracks are dashed. For motorcycle 1, the confirmed track
and the tentative track go to $\chi_{12}=c_1=t_1=\kappa_1$.
The stack $S_1$ only records $d_1$. For motorcycle 2,
the confirmed track ends at $c_2=\chi_{12}$, the tentative
track ends at $t_2$, and
$S_2=(t_2,\chi_{23},d_2)$. For motorcycle 3, we have
$c_3=\chi_{23}$, $t_3=d_3$, and $S_3=(d_3)$. For motorcycle
4, we have $c_4=s_4$, $t_4=d_4$, and $S_4=(d_4)$.
\label{fig:tracks}}
\end{figure}
The stack $S_i$ is ordered from $c_i$ to $d_i$.
We denote by $\Top(S_i)$ its first element, so $\Top(S_i)$ is
the target point in $S_i$ that is closest to $c_i$.
At the beginning, we set $S_i=\{s_i,d_i\}$ for all $i$.
New target points will be created in Case (3b) of our algorithm, 
as described below.

If motorcycle $i$ has neither crashed nor stopped, 
then its tentative track ends at the first target point in $S_i$, 
so $t_i=\Top(S_i)$.  Otherwise, the tentative track and the confirmed
track are the same, thus $t_i=c_i$. So after a motorcycle has
crashed or stopped, the ray shooting data structure records its
confirmed track.

An {\em event} $(i,p)$ happens when a motorcycle $i$ reaches a target 
point $p$. We process events one by one, and while an event is being
processed, new events may be generated. After an event has been
processed, we process the earliest available event. As $t_i=\Top(S_i)$ is
the closest target point to $i$ in $S_i$, it means that we always
process the event $(i,t_i)$ such that $\tau(i,t_i)$ is smallest.
Note that it does not imply that our simulation is done in chronological order:
When we process an event $(i,t_i)$, we may create a new event $(j,p)$ such
that $\tau(j,p) < \tau(i,t_i)$. (See Appendix~\ref{sec:example}.)

We record in a priority queue $\mathcal Q$ 
the event $(i,t_i)$ for each motorcycle $i$ that has
not crashed or stopped. An event with earlier 
time $\tau(i,t_i)$ has higher priority. As $t_i=\Top(S_i)$,
we can update the event queue $\mathcal Q$ in $O(\log n)$ time each time
a stack $S_i$ is updated.
So we can find the next available event in
$O(\log n)$ time. The first $n$ events are the events 
$(i,s_i)$, $i=1,\dots,n$, and occur at time $t=0$.  We process these $n$ events
in an arbitrary order.

We now explain how to process an event $(i,t_i)$. To avoid confusion, for any
motorcycle $j$, we use the notation $c_j,t_j$ to denote the endpoints of
its confirmed and tentative track just before processing this event, and
we use the notation $c'_j,t'_j$ for their position just after processing this
event.
We first extend the confirmed track of motorcycle $i$ to $t_i$, 
thus $c'_i= t_i$. We also delete $t_i$ from $S_i$. We are now in 
one of the following cases:
\begin{itemize}
\item[(1)] If $t_i=d_i$, then motorcycle $i$ stops.
In order to avoid processing irrelevant events in the future, 
we remove $S_i$ from $\mathcal Q$. 
\item[(2)] If $t_i$ is a crossing point $\chi_{ij}$ that lies in
the confirmed track of $j$ (that is, $t_i \in \overline{s_jc_j}$), then 
$i$ crashes at $t_i$. So we remove $S_i$ from $\mathcal Q$.
\item[(3)] Otherwise, we try to extend the tentative track to the next target
point $q=\Top(S_i)$. So we perform a ray shooting query with ray $(t_i,\vec v_i)$,
which gives us the first track intersected by $\overline{t_iq}$, if any.
\begin{itemize}
\item[(3a)] If $\overline{t_iq}$  does not cross any track, then 
$t'_i=q$, and we do not need to do anything else to handle this event.
\item[(3b)] Otherwise, let $j$ be the result of the 
ray-shooting query, so $\overline{s_jt_j}$ 
is the first track hit by segment $\overline{t_iq}$, starting from $t_i$.
We shorten the tentative track of $i$, which means that we insert the new target 
point $\chi_{ij}$ into $S_i$, as well as the point $t'_i=h(t_i,\chi_{ij})$ 
obtained by a halving query on $\overline{t_i\chi_{ij}}$.
If the crossing point $\chi_{ij}$ does not lie in the confirmed track 
of $j$, that is, if $\chi_{ij} \in \overline{c_jt_j} \setminus \{c_j\}$, then
we also shorten the tentative track of $j$, so we insert $\chi_{ij}$ into 
$S_j$, and we insert $t'_j=h(c_j,\chi_{ij})$ into $S_j$.
\end{itemize}
\end{itemize}
After applying the rules above, we update the ray shooting 
data structure (if needed), and
we move to the next available event.

\subsection{Proof of correctness}\label{sec:correctness}

Initially, we create the target points $s_i,d_i$ for $i=1,\dots,n$. 
After this, we create new target points only in Case (3b) of our algorithm.
There are two types of such target points: the crossing
points $\chi_{ij}$ obtained by ray-shooting, and the points obtained by 
halving queries. We call $\chi$-targets the first type of target points, 
and $h$-targets the latter. By our assumption that the result of a
halving query is not a crossing point, a target point cannot
be both a $h$-target and a $\chi$-target. 

We need the following lemma. Remember that we say that two
segments cross if their relative interiors intersect.
\begin{lemma}\label{lem:invariant}
During the course of the algorithm, no two tentative tracks cross.
\end{lemma}
\begin{proof}
For sake of contradiction, assume that two tentative tracks cross 
during the course of the algorithm. 
Let $(i,t_i)$ be the first event that generates such a
crossing, so just before processing this event, the tentative
tracks $\overline{s_jt_j}$, $j=1,\dots,n$ do not cross, and 
there is a crossing among the tracks $\overline{s_jt'_j}$.
We must be in Case (3), because we do not extend any
tentative track in cases (1) and (2).
Besides, we only extend the track of motorcycle $i$ in Case (3).
So there must be another motorcycle $k \neq i$ such that
$\overline{s_it'_i}$ crosses $\overline{s_kt_k}$.
 
In Case (3a), the segment
$\overline{t_iq}$ obtained by ray shooting does not cross any
tentative track $\overline{s_jt_j}$, $j\neq i$, 
and since $t'_i=q$, then the new portion $\overline{t_it'_i}$
of the track does not cross any other tentative track. The same
is true in Case (3b), because $t'_i$ is in $\overline{t_i\chi_{ij}}$,
where track $j$ is the first track hit by $\overline{t_iq}$.
So we just proved that, in any case, the new portion $\overline{t_it'_i}$ 
of the track does not cross any track $\overline{s_jt_j}$, $j\neq i$,
and in particular, $\overline{t_it'_i}$ does not cross $\overline{s_kt_k}$.

By our assumption, we also know that $\overline{s_kt_k}$ cannot
cross $\overline{s_it_i}$. So the only remaining possibility is that
$\overline{s_kt_k}$ crosses $\overline{s_it_i'}$ at $t_i$.
Then $t_i$ is the crossing point $\chi_{ik}$. 
This point $t_i=\chi_{ik}$ cannot be in the confirmed track $\overline{s_kc_k}$,
because that would be Case (2) of our algorithm, and we showed that 
we are in Case (3).  
Since $\chi_{ik}$ is a $\chi$-target of $i$, 
and it does not lie in the confirmed track of $k$,
then it must have been inserted at the same time
in $S_i$ and $S_k$ while processing a previous event.
Since $\chi_{ik}$ is not on the confirmed track of $k$,
then it must still be in $S_k$. 
So the  tentative track $\overline{s_kt_k}$
cannot contain $\chi_{ik}$ in its relative interior, a contradiction.
\end{proof}

We want to argue that our algorithm computes the motorcycle
graph correctly. So assume it is not the case. As our algorithm moves
motorcycles forward until they either reach their destination point or crash, 
it could only fail if during the execution of our algorithm, 
the confirmed track of at least one motorcycle $i$ goes beyond the  
point $\kappa_i$  where it is supposed to crash in the motorcycle graph. 
Let us consider the event $(i,t_i)$ that is first 
processed by our algorithm, such that motorcycle $i$ goes beyond $\kappa_i$.
So $\kappa_i$ is in the segment $\overline{c_it_i} \setminus \{t_i\}$.
Let $j$ denote the motorcycle that $i$ crashes into, in the 
(correct) motorcycle graph, so $\kappa_i=\chi_{ij}$.

When we process $(i,t_i)$, in the current graph constructed by
our algorithm, motorcycle $j$ cannot have reached $d_j$, because
it would mean that the tentative tracks $\overline{s_it_i}$ and
$\overline{s_jd_j}$ are crossing at $\kappa_i=\chi_{ij}$, which is impossible
by Lemma~\ref{lem:invariant}.

We now rule out the case where, when our algorithm processes
$(i,t_i)$, motorcycle $j$ has already crashed into some motorcycle $k$ in
the graph constructed by the algorithm. (See \figurename~\ref{fig:correctness2}.)
For sake of contradiction, assume it did happen. 
\begin{figure}
\begin{center}
\includegraphics{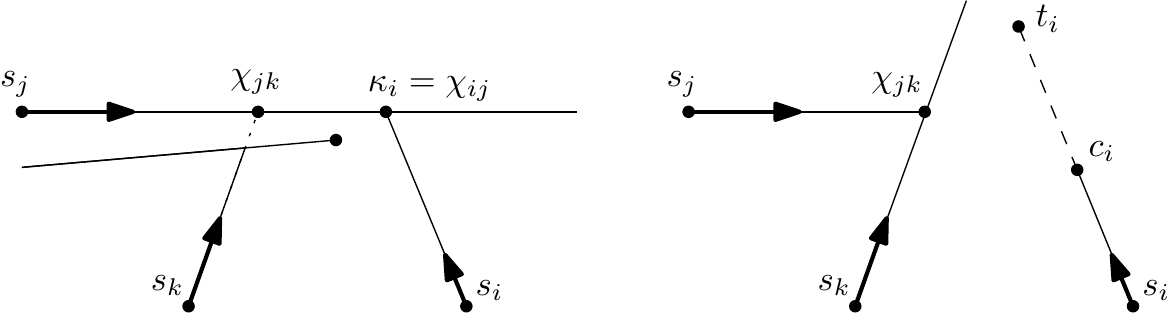}
\end{center}
\caption{The motorcycle graph (left) and an incorrect computation (right).
\label{fig:correctness2}}
\end{figure}
\begin{itemize}
\item If we had $i=k$, then $\chi_{ij}$ would have been created
as a $\chi$-target for $j$ earlier. At this point, $i$
had not gone past $\chi_{ij}$, because $(i,t_i)$ is the first
such event. As $\tau(j,\chi_{ij})<\tau(i,\chi_{ij})$, the
algorithm would have moved $j$ to $\chi_{ij}$ before $i$
moves further, and thus $j$ would not crash at $\chi_{ij}$,
a contradiction.
\item Thus we must have $i \neq k$. 
As $t_i$ is beyond $\kappa_i=\chi_{ij}$,
and tentative tracks cannot cross, we must have 
$c_j \in \overline{s_j \chi_{ij}}$. So $j$ crashed into $k$ at
$\chi_{jk} \in \overline{s_j\chi_{ij}}$. As in the correct
motorcycle graph, $j$ does not crash into $k$, it means that
the algorithm has already moved $k$ past its (correct) crashing
point, which contradicts our assumption that $(i,t_i)$ was the
first such event.
\end{itemize}

We just proved that $j$ has not stopped or crashed when
the algorithm processes event $(i,t_i)$, so at this point
there should be an event $(j,t_j)$ in the queue.
By Lemma~\ref{lem:invariant}, the tracks $\overline{s_it_i}$ 
and $\overline{s_jt_j}$
cannot cross, so we must have $t_j \in \overline{c_j\kappa_i}$. 
(See \figurename~\ref{fig:correctness}.) 
\begin{figure}
\begin{center}
\includegraphics{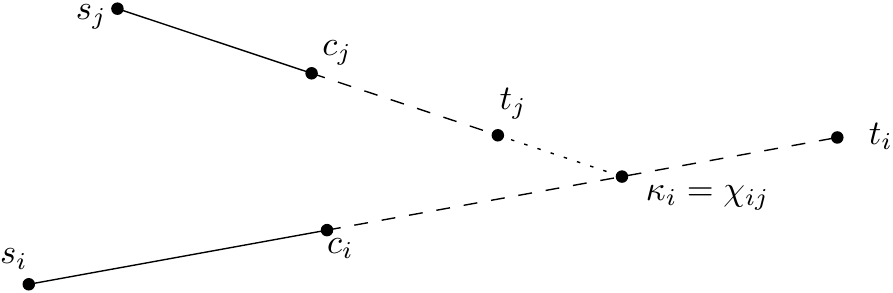}
\end{center}
\caption{Proof of correctness, remaining case.\label{fig:correctness}}
\end{figure}
It implies that 
$\tau(j,t_j) \leq \tau(j,\kappa_i)$. But since $i$ crashes 
into $j$ in the (correct)
motorcycle graph, we must have $\tau(j,\kappa_i)<\tau(i,\kappa_i)$, thus
$\tau(j,t_j) < \tau(i,\kappa_i)$. As $\kappa_i \in \overline{c_it_i}$, we have
$\tau(i,\kappa_i)\leq \tau(i,t_i)$, thus $\tau(j,t_j) < \tau(i,t_i)$.
But this is impossible, because our algorithm always processes the
earliest available event, so it would have processed 
$(j,t_j)$ rather than $(i,t_i)$.

\subsection{Analysis}

Our algorithm uses two auxiliary data structures: for answering halving 
queries, and for ray shooting. The running time of our algorithm depends
on their preprocessing time and query time. Let $P(n)$ denote an
upper bound on the preprocessing
time of these two data structures, and let $Q(n)$ denote an upper
bound on the time needed for a 
query or update---so we can answer a ray-shooting query or a halving 
query in time $Q(n)$, and we can update the ray shooting
data structure in time $Q(n)$. We now prove the following result:
\begin{theorem}\label{th:analysis}
We can compute a motorcycle graph of size $n$ in time 
$O(P(n)+n(Q(n)+\log n)\log n)$.
\end{theorem}

Each time we handle an event, we perform at most two halving queries, 
one ray-shooting
query, and we may update two tentative tracks in the ray-shooting data
structure. We also pay an $O(\log n)$ time overhead to update 
the priority queue $\mathcal Q$. So after preprocessing, the running time will 
be at most  the number of events times $Q(n)+\log n$. Thus we only need to
argue that our algorithm processes a total of $O(n \log n)$ events. In fact,
at each event we process, a motorcycle reaches a target point, so we only 
need to show
that $O(n \log n)$ target points are created during the course of the algorithm.

Initially, we create $O(n)$ target points, which are $s_i,d_i$  for $i=1,\dots,n$.
After this, we only  create new target points in Case (3b) of the algorithm.
In this case, we create one $\chi$-target,
and at most two $h$-targets obtained by halving. Thus we only
need to bound the number of $\chi$-targets. At the end of the algorithm, some of
these $\chi$-targets $\chi_{ij}$ correspond to an actual crash, with motorcycle $i$
crashing into $j$, or $j$ crashing into $i$. In any case, there are at most
$n$ such $\chi$-targets. We need to consider the other $\chi$-targets, that do
not correspond to an actual crash. In this case, either motorcycle $i$
or $j$ does not reach $\chi_{ij}$, so at the end of the computation, 
$\chi_{ij}$ must appear in the stack $S_i$
or $S_j$ of target points that have not been reached by motorcycle $i$ or
$j$, respectively. Thus, in order
to complete the proof of Theorem~\ref{th:analysis}, we only need the following
lemma.
\begin{lemma}\label{lem:targets}
At the end of the execution of our algorithm, for any motorcycle $i$,
the number of $\chi$-targets in $S_i$ is $O(\log n)$.
\end{lemma}
\begin{proof}
In this proof, we only consider the status of the stack $S_i$
at the end of the algorithm, and we assume that it contains more
than one $\chi$-target. We denote by $\chi_1,\dots,\chi_m$
the $\chi$-targets in $S_i$, in reverse order, so
$\chi_m\dots\chi_2\chi_1$ is a subsequence of $S_i$,
where $\chi_m$ is closest to $\kappa_i$ and $\chi_1$ is
closest to $d_i$.

Each target $\chi_j$ was created in case (3b) of our
algorithm. At the same time, an $h$-target $h_j=h(g_j,\chi_j)$ was created
by a halving query using another target point $g_j$.
As the points $\chi_j$, $j=1,\dots,m$ are in $S_i$, motorcycle $i$
never reaches these points during the course of the algorithm,
so $\chi_1$ and $h_1$ must have been created first, then 
$\chi_2$ and $h_2$ \dots and finally $\chi_m$ and $h_m$.

For any $2\leq j \leq m$, as $\chi_j$ is created after $\chi_{j-1}$,
and these two points are created when motorcycle $i$ reaches
$g_j$ and $g_{j-1}$, respectively, it implies that $g_{j-1}$ is
in $\overline{s_ig_{j}}$. We also know that 
$\chi_{j-1}$ lies in $\overline{\chi_jd_i}$, because $\chi_{j-1}$
appears after $\chi_j$ in $S_i$. So $\overline{g_j\chi_j}$, $j=1,\dots,m$ 
is a sequence of nested segments, that is, we have 
$\overline{g_j\chi_j} \subset \overline{g_{j-1}\chi_{j-1}}$
for all $2 \leq j \leq m$. More precisely:
\begin{itemize}
\item If $h_{j-1}$ is in $S_i$, 
then $\overline{g_j\chi_j} \subset \overline{g_{j-1}h_{j-1}}$,
because $\chi_j$ is created after $h_{j-1}$, and
motorcycle $i$ never reaches $h_{j-1}$. 
(See \figurename~\ref{fig:targets1}.)
\item If $h_{j-1}$ is not in $S_i$, then 
$\overline{g_j\chi_j} \subset \overline{h_{j-1}\chi_{j-1}}$.
(See \figurename~\ref{fig:targets2}.)
It can be proved as follows. As $h_{j-1}$ is created at
the same time as $\chi_{j-1}$, then $\chi_j$ is created
after $h_{j-1}$. So $\chi_j$ must have been created
after motorcycle $i$ reaches $h_{j-1}$, otherwise
we would have $\chi_j \in \overline{s_ih_{j-1}}$,
and since motorcycle $i$ reaches $h_{j-1}$ later,
$\chi_j$ would not be in $S_i$.
As $\chi_j$ is created after motorcycle $i$ reaches $h_{j-1}$,
we must have $g_j \in \overline{h_{j-1}\chi_{j-1}}$.
\end{itemize}
\begin{figure}
\begin{center}
\includegraphics[width=.9\textwidth]{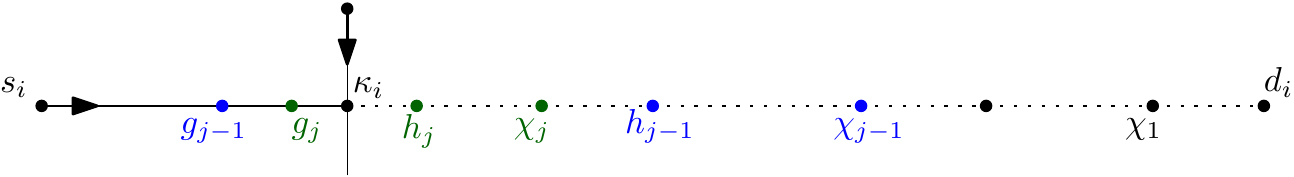}
\end{center}
\caption{Proof of Lemma~\ref{lem:targets}, first case.\label{fig:targets1}}
\end{figure}
\begin{figure}
\begin{center}
\includegraphics[width=.9\textwidth]{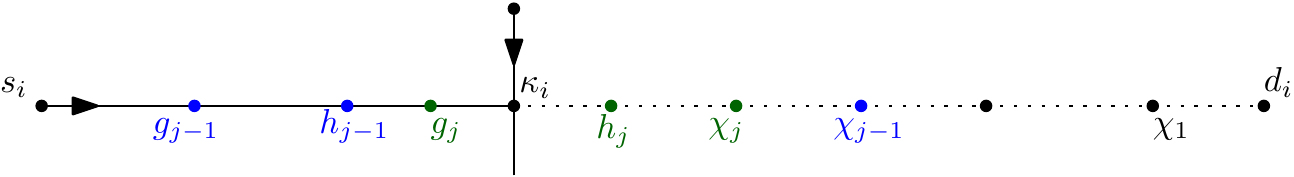}
\end{center}
\caption{Proof of Lemma~\ref{lem:targets}, second case.\label{fig:targets2}}
\end{figure}

Thus $\overline{g_j\chi_j}$ is contained in
either $\overline{g_{j-1}h_{j-1}}$ or $\overline{h_{j-1}\chi_{j-1}}$,
and since $h_{j-1}=h(g_{j-1},\chi_{j-1})$, it follows that the
size $|g_j\chi_j|$ decreases exponentially when $j$ increases from 1 to $m$. As
$\overline{g_{m-1}\chi_{m-1}}$ contains $\chi_{m}$ and $\chi_{m-1}$,
we have $|g_{m-1}\chi_{m-1}| \geq 2$. In addition,
$|g_1\chi_1|\leq n$, so we must have $m=O(\log n)$.

\end{proof}

\section{Auxiliary Data Structures}

Our algorithm, presented in Section~\ref{sec:algorithm}, requires
two auxiliary structures. The first one is simply a ray-shooting
data structure. As ray shooting is a standard operation in
computational geometry, we will be able to directly use known
data structures. The second data structure we need is for answering 
halving queries. We show below how to construct efficient data 
structures for this type of queries, and the corresponding time
bounds for our motorcycle graph algorithm. 

\subsection{General case}\label{sec:general} 
In this section, we present the auxiliary data structures
for the most general case, as presented in Section~\ref{sec:notation}.
So motorcycles have arbitrary starting position, destination point and
velocity.

For ray shooting, we can directly use a data structure by Agarwal
and Matou{\v{s}}ek~\cite{agarwal1993ray}, which requires preprocessing
time $O(n^{4/3+\eps})$, with update and query time $O(n^{1/3+\eps})$, for any
$\eps>0$.

For halving queries, we use known range
searching data structures and parametric search, as in the work
of Agarwal and Matou{\v{s}}ek on ray shooting: Our problem is  
an optimization version of range counting in an arrangement of lines, 
so we obtain the same bounds~\cite[Section 3.1]{agarwal1993ray}.
\begin{lemma}\label{lem:halving}
Given the $n$ supporting lines $\ell_1,\dots,\ell_n$, we can construct
a data structure with $O(n^{4/3+\eps})$ preprocessing time
and $O(n^{1/3+\eps})$ query time that answers the following queries $(i,p,q)$.
Assume there are $k$ crossing points $\chi_{ij}$
on $\overline{pq}$. Then we return the median crossing point and
the next: the
$\lceil k/2\rceil$th and the $(\lceil k/2\rceil+1)$th such crossing
point, in the ordering from $p$ to $q$ along $\overline{pq}$.
\end{lemma}

With the two auxiliary data structures above, Theorem~\ref{th:analysis}
yields the following result.
\begin{theorem}\label{th:general}
A motorcycle graph can be computed in $O(n^{4/3+\eps})$ time, 
for any $\eps>0$.
\end{theorem}

It should be possible to replace the $n^\eps$ factor in the bounds
of Lemma~\ref{lem:halving} with a polylogarithm using known
range searching techniques~\cite{Chan12,Matousek92}, because we
only need a static data structure for halving queries, but in any
case we need a dynamic data structure for ray shooting queries, so
it would not improve our overall time bounds.

\subsection{$C$-Oriented Motorcycle Graphs}\label{sec:coriented}

We consider the special case where motorcycles can only take $C$
different directions $\vec{d_1},\dots,\vec{d_C}$.
Eppstein and Erickson gave an $O(n^{4/3+\eps})$-time
algorithm when $C=O(1)$. We show that with appropriate auxiliary data
structures, we can solve this case in time $O(n\log^3 n)$. In the
following, we do not assume that $C=O(1)$, so our time bounds will
also have a dependency on $C$.

\begin{proposition}\label{cor:coriented}
We can compute a $C$-oriented motorcycle graph in
$O(Cn\log^2(n) \min(C,\log n))$ time.
\end{proposition}

We use the following data structures, and then the result
follows from Theorem~\ref{th:analysis}.

\paragraph{Ray shooting data structures.} A first approach to answer
our ray shooting queries is to use $C$ instances of a data structure
for vertical ray shooting in a planar subdivision. Several
data structures are known for this problem~\cite{ArgeBG06}, we
use a data structure by Cheng and Janardan~\cite{ChengJ92} that takes
$O(\log^2 n)$ time per update and $O(\log n)$ time per query.
So overall, we get $Q(n)=O(C \log^2 n)$ with the terminology of
Theorem~\ref{th:analysis}.

Alternatively, we can use $C(C-1)$ instances of a data structure
for vertical ray shooting among horizontal segments. Each
data structure is used to answer ray shooting queries with
a given direction, into segments with another direction:
We just need to change the two coordinate axis to these two directions.
Using a recent result by Giyora and Kaplan~\cite{GiyoraK07}, we obtain
$Q(n)=O(C^2\log n)$.

\paragraph{Halving queries.} Our data structure for halving queries
simply consists of a sorted list of motorcycles for each direction.
So for each $k \in 1,\dots,C$, we have an array $\mathcal A_k$ of the
motorcycles with direction $\vec d_k$, sorted according to the intercept
of their supporting lines with a line orthogonal to $\vec d_k$.
We now explain how to answer a halving query $\overline{pq} \subset \ell_i$.

Without loss of generality, assume $\ell_i$ has direction $\vec d_1$. For each
direction $\vec d_k$, $k=2,\dots,C$, the subset of motorcycles whose supporting
lines cross $\overline{pq}$ appear in consecutive positions in $\mathcal A_k$.
We can find the first and the last index of these lines in $O(\log n)$ time
by binary search. So we obtain all the arrangement vertices in $\overline{pq}$
in $C-1$ sorted subarrays. We then compute the median $m_k$ of each 
such subarray $\mathcal A_k \cap \overline{pq}$,
and the median of these points $m_k$ weighted by the number of points
$|\mathcal A_k \cap \overline{pq}|$ in the corresponding subarray. This gives
a halving point $h(p,q)$ with $\rho=3/4$. The median of each subarray can
be found in $O(1)$ time, and their weighted median in $O(C)$ time~\cite{CLRS},
so the query time is dominated by the $C$ binary searches. Thus, we can
answer halving queries in $O(C\log n)$ time.

\subsection{Bounded precision input}\label{sec:bpmg}

The data structure for answering halving queries in Section~\ref{sec:general}
is quite involved. In practice, one would rather implement halving
queries by simply halving the Euclidean length $\|pq\|$
instead of approximately halving the number of crossing points. Unfortunately, 
in the infinite precision model that is commonly used in computational 
geometry, this would cause
the analysis of our algorithm in Lemma~\ref{lem:targets} to break down,
because a stack of target points $S_i$ may have size $\Omega(n)$
at the end of the algorithm. 

Such a counterexample would require the distance
between consecutive target points in $S_i$ to become exponentially small 
near the crashing point, which does not seem likely to happen in practice.
To formalize this idea, we make the assumption
that all input numbers (the coordinates of the starting points, the 
destination points, and the velocities) are rational numbers,
whose numerator and denominator are in $\{-2^{w-1},\dots,2^{w-1}-1\}$ for some 
integer $w$. In other words, the input numbers are $w$-bit signed integers.
We still assume that arithmetic operations between two numbers can be performed 
in constant time.  

This model also allows us to handle the case
where input coordinates are $w$-bit rational numbers,
that is, rational numbers with $w$-bit numerator
and denominator; we just need to scale up each coordinate by a
factor $2^{w}$ to obtain $2w$-bit integers, losing only
a constant factor in our time bounds.
In the proofs below we assume the input numbers are integers, to
simplify the presentation, but the results are stated with rational
coordinates.

As the input coordinates are $w$-bit integers, 
the coordinates of a crossing point $\chi_{ij}$ are rational
numbers obtained by solving a $2 \times 2$ linear system,  
their denominator being the determinant  $\det(\vec v_i,\vec v_j)$.
Thus, the denominator is an integer between $-2^{2w-1}$ and $2^{2w-1}$.
So any two distinct crossing points are at distance at
least $2^{-2w+1}$ from each other. 

Assume that we replace our halving operation, as defined in 
Section~\ref{sec:notation}, with halving the Euclidean
length. So $h(p,q)$ is the midpoint of $\overline{pq}$,
which can be computed in constant time.
Then any nested sequence of segments obtained by
successive halving, as in the proof of Lemma~\ref{lem:targets},
consists of $O(w)$ such nested segments, because
a segment $\overline{g_j\chi_j}$ of length smaller than
$2^{-2w+1}$ cannot contain another crossing point in its
interior, and hence it cannot be further subdivided.
So the bound on the size of $S_i$ becomes $O(w)$, and we get the following 
result.

\begin{theorem}\label{th:bpmotorcycle}
If the input coordinates are $w$-bit rational numbers, we can
compute a motorcycle graph in time  $O(nw(Q'(n)+\log n))$,
where $Q'(n)$ is the time needed for updates or queries in
the ray-shooting data structure.
\end{theorem}

For bounded-precision input, the bottleneck
of our algorithm has thus become the ray shooting data structure,
whose update and query time bound is $O(n^{1/3+\eps})$ in 
the most general case. Therefore, we obtain a faster
motorcycle graph algorithm if we are in a special case where
faster ray-shooting data structures are known. One such
case is ray-shooting in a connected
planar subdivision, which can be done in $O(\log^2 n)$-time
per update and query using a data structure by Goodrich
and Tamassia~\cite{goodrich1997dynamic}. We can use
this data structure if, for instance, all motorcycles move
inside a simple polygon $P$, starting from its boundary. 
(So for all $i$, $\overline{s_id_i} \subset P$, and $s_i$ is on the boundary
of $P$.) Then we perform ray shooting in the union of the tentative
tracks and the edges of $P$, which  form a connected subdivision.
It yields the following time bound.
\begin{corollary}\label{th:bpsmotorcycle}
We can compute a motorcycle graph in time $O(n\log^3 n)$ 
for $n$ motorcycles moving inside a simple polygon
with $O(n)$ vertices, starting on its boundary, and if
the input has $O(\log n)$-bit rational coordinates.
\end{corollary}

\section{Application to Straight Skeleton Computation}\label{sec:skeleton}

In this section, we give new results on straight skeleton 
computation, using our new motorcycle graph algorithm.

\subsection{Preliminaries and non-degenerate cases}\label{sec:prelss}

As we mentioned in the introduction, the straight skeleton problem
and the motorcycle graph problem are closely related. We now explain
it in more details.

Consider the reflex (non-convex) vertices of a polygon $P$. When we 
construct the straight skeleton of $P$, these vertices move inward and may
collide into edges, or other vertices. These events, called 
split events and vertex events, are the difficult part of straight
skeleton computation, because they affect the topology of the
shrinking polygon by splitting it, and because they are non-local: 
A reflex vertex may affect a chain of edges on the other side of 
the polygon. The other type of events, called edge events, where
an edge shrinks to a point, are easily handled with a priority queue.
So the interaction between reflex vertices is a crucial part
in straight skeleton computation, and the motorcycle graph
presented below helps to determine these interactions.

The motorcycle graph {\em induced} by a polygon $P$ is such that 
each motorcycle starts at a reflex vertex, 
moves as the same velocity as the corresponding reflex vertex when 
we shrink $P$, and stops if it reaches the boundary $\partial P$ of $P$.
(See \figurename~\ref{fig:mgss}.)

If $P$ is degenerate, then two reflex vertices may collide and
create a new reflex vertex. In this case we need to create
a new motorcycle after the collision. 
(See \figurename~\ref{fig:degenerate}.)
So when two motorcycles
collide in the induced motorcycle graph, we may have to create
a new motorcycle~\cite{HuberH11}. Our motorcycle
graph algorithm, as described above, does not apply directly
to this case, because the proof of Lemma~\ref{lem:targets}
breaks down. (The reason is that $S_i$ may hold a linear number of target points
at the end of the execution of the algorithm, due to the newly created motorcycles.
See \figurename~\ref{fig:counterexample}.)
\begin{figure}
\begin{center}
\includegraphics[width=.9\textwidth]{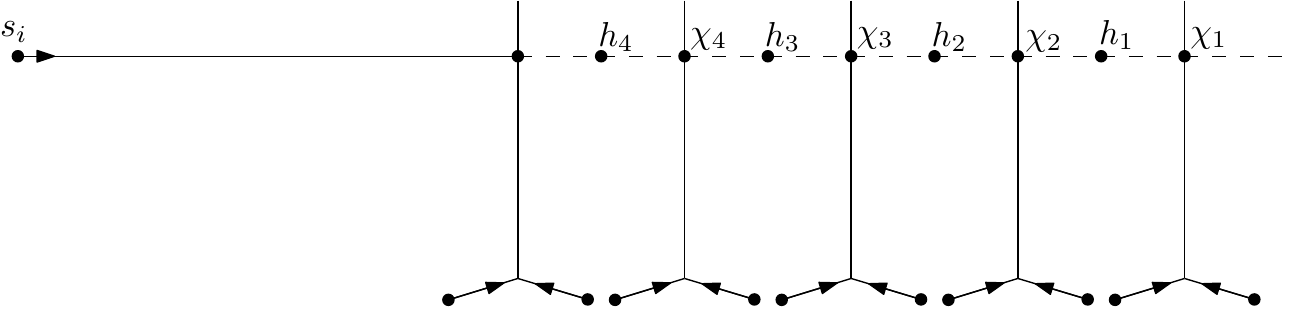}
\end{center}
\caption{An example where $S_i$ holds a linear number of target points
at the end of execution of the algorithm. The speed of the motorcycle
at the bottom are adjusted so that $\chi_1$ is created first, then $\chi_2$\dots
\label{fig:counterexample}}
\end{figure}
In Section~\ref{sec:bpss}, we  will explain how to compute these generalized motorcycle 
graphs  efficiently on bounded-precision input.
But the following theorem still holds in degenerate cases.

\begin{theorem}[Cheng and Vigneron~\cite{ChengV07}, 
Huber and Held~\cite{HuberH11}]\label{th:reduction}
The straight skeleton of a polygon $P$ with $n$ vertices and
$h$ holes can be computed in expected time $O(n\sqrt{h+1}\log^2 n)$
if we know the motorcycle graph induced by the vertices of $P$.
\end{theorem}

From the discussion above, and using our motorcycle graph algorithm
from Theorem~\ref{th:general}, we obtain the following result.
\begin{corollary}\label{cor:motorgeneral}
We can compute the straight skeleton of a non-degenerate polygon
with $n$ vertices and $h$ holes in
$O(n^{4/3+\eps}+n\sqrt{h+1}\log^2 n)$ time for any $\eps>0$.
\end{corollary}

\subsection{Bounded precision input}\label{sec:bpss}

We use the same bounded precision assumptions as in Section~\ref{sec:bpmg},
where the input coordinates are $w$-bit integers or, equivalently,
$w$-bit rational numbers. Similarly, to simplify the presentation,
we use the integer model in the proofs, but we state the results
in the rational model.

Thus, the coordinates of the vertices of the input polygon $P$ are
$w$-bit integers. The supporting lines $\ell_i$ of the motorcycles
are angle bisectors between two edges of the input polygon.
In order to apply the same halving scheme as in Section~\ref{sec:bpmg},
where the Euclidean length is used instead of the number of arrangement
vertices, we need to argue that the separation between two 
vertices in this arrangement of bisectors cannot be too small.
This distance can be shown to be at least
$2^{-W}$, where $W=64(80w+105)+1=O(w)$, by applying the
separation bound by Burnikel et al.~\cite{burnikel1997strong}.
So we obtain a result for induced motorcycle graphs that is
analogous to Theorem~\ref{th:bpmotorcycle}.
\begin{lemma}\label{lem:bpinduced}
Given a polygon $P$ whose input coordinates are $w$-bit rational numbers, 
we can compute the motorcycle graph induced by $P$ in time  $O(nw(Q'(n)+\log n))$,
where $Q'(n)$ is the time needed for updates or queries in
the ray-shooting data structure.
\end{lemma}
In the lemma statement above, we do not exclude degenerate cases.
This is another advantage of this bounded precision model. As
the argument in our analysis only relies on the separation bound
between two distinct crossing points, and not on the number of 
motorcycles crossing a given segment, a newly created motorcycle
does not affect our analysis as it still obeys the same separation
bound: A newly created motorcycle still lies on the bisector
of two input edges, though these two edges are not adjacent
in the input polygon~\cite{HuberH11}. 
(See \figurename~\ref{fig:degenerate}.)

We still need to describe an efficient ray-shooting data structure.
As our input polygon has $h$ holes, the boundary $\partial P$ of $P$
together with the tentative tracks form a collection of $h+1$ disjoint
simple polygons. We could directly use known ray-shooting data 
structures~\cite{AgarwalS96a,HershbergerS95}, which can be made
dynamic at the expense of an extra $n^\eps$ factor in the running
time~\cite{agarwal_geometric_1998}. In the following, we give
a different approach, which leads to a better time bound when
used as a subroutine of our algorithm.
This approach takes advantage of the fact that
the holes of $P$ are fixed (only the tentative tracks are dynamic).
We use a spanning tree with low crossing number, which is not
a new idea in ray-shooting data 
structures~\cite{ChazelleEGGHSS94,HershbergerS95}.

We pick one point on the boundary of each hole of $P$, and on the 
boundary of $P$. We connect these $h+1$ points using a spanning tree 
$\mathcal T$ with low stabbing number~\cite{agarwal1993applications}, that
is, a spanning tree such that any line crosses at most $O(\sqrt{h})$ 
edges of $T$. This tree can be computed in $O(n^{1+\eps})$ 
time~\cite[Section 8]{agarwal1993applications}.
We maintain a polygonal subdivision which is the overlay of
$P$ with  $\mathcal T$ and the tentative tracks.
So at each intersection between an edge of $\mathcal T$ and an
edge of $P$ or a tentative track, we split the corresponding
edges and tracks at the intersection point. This subdivision $\mathcal S$ 
is connected and has $O(n\sqrt{h})$ edges, and we maintain it in the
ray shooting data structure by Goodrich and 
Tamassia~\cite{goodrich1997dynamic}, which has $O(\log^2 n)$ 
update and query time. 

 Each time a tentative track
is extended or retracted, as a tentative track intersects $O(\sqrt{h})$
edge of $\mathcal T$, we can update the subdivision and the data
structure by making $O(\sqrt{h})$ updates in the ray
shooting data structure. Similarly, when our motorcycle graph
algorithm tries to extend a tentative track, we can find the
first tentative track being hit by a query ray in $O(\sqrt{h} \log^2 n)$
time: We first perform a ray shooting query in $\mathcal S$, which
takes $O(\log^2 n)$ time. If we hit an edge of $\mathcal T$, we make
another ray shooting query starting at the hitting point of the previous
query, and in the same direction. We repeat this process as long as
the result of the query is an edge of $\mathcal T$, and by the low-stabbing
number property, it may only happen $O(\sqrt{h})$ times.

Overall, our ray shooting data structure has update and query time
$O(\sqrt{h}\log^2 n)$. So by Theorem~\ref{th:reduction} and
Lemma~\ref{lem:bpinduced}, we obtain the following result. Note that
it still holds for degenerate input.
\begin{theorem}\label{th:bpss}
The straight skeleton of a polygon with $n$ vertices and $h$ holes, 
whose coordinates are $O(\log n)$-bit rational numbers, can be computed
in $O(n\sqrt{h+1}\log^3 n)$ expected time. 
\end{theorem}

\pagebreak
\begin{center}
\huge {\bf Appendix} \normalsize
\end{center}

\appendix

\section{Pseudocode}\label{sec:pseudocode}

In this section, we give the pseudocode of our algorithm. It is more
detailed than the algorithm description in Section \ref{sec:algorithm}, and
it can handle degenerate cases. The proof of correctness and the analysis are
essentially the same as in Section~\ref{sec:main}, but they require
a more detailed case analysis.

To deal with the degenerate cases where some supporting lines
are concurrent, or two or more motorcycles reach a point at the same
time, we record all the target points created so far in a 
dictionary data structure $\mathcal D$. We can implement $\mathcal D$
as a balanced binary search tree, sorted in lexicographical order of
the coordinates $(x,y)$,
which allows to retrieve a point in $O(\log n)$ time.
We associate three fields with each point $p$ stored in $\mathcal D$:
\begin{itemize}
\item A list $M(p)$ records the motorcycles $i$ such that 
$p \in S_i$. So $M(p)$ records
all the motorcycles $i$ that could possibly reach $p$, at a given
point of the execution of our algorithm. The set $M(p)$ itself
is stored in a dictionary data structure, so that we can decide
in $O(\log n)$ time whether a motorcycle $i$ is in $M(p)$.
\item Two motorcycles $k,k'$ of $M(p)$ such that
$\tau(k,p),\tau(k',p)$ are smallest. It will allow us to
find out whether two motorcycles crash simultaneously at $p$.
\item A flag $\blocked(p)$ which is set to FALSE initially,
and is set to TRUE as soon as a confirmed track has reached $p$, 
implying that any other motorcycle that reaches $p$ must crash.
\end{itemize}

After the initialization stage, our algorithm handles repeatedly
the earliest available event, according
to the four cases (1), (2), (3a) and (3b) from 
Section~\ref{sec:algorithm}.

Lines~\ref{line:check_0} and~\ref{line:check_1} deal with Case (1) and (2). 
The condition $p=d_i$ corresponds to Case (1). The other two conditions check
whether we are in Case (2). 
In particular, condition $\blocked(p)=\mbox{TRUE}$ means that at least one other 
motorcycle has reached $p$, thus motorcycle $i$ crashes. 
With degenerate input, it is possible that another (or several other)
motorcycle reaches $p$ at the same time as $i$, 
in which case $\blocked(p)=\mbox{FALSE}$ if $(i,p)$ is the 
first event involving $p$ that has been processed.
The condition at Line~\ref{line:check_1} checks whether we are
in this situation. If so, $i$ must crash.  
As $(i,p)$ is the first event involving $p$ that we process,
there is no earlier event in $M(p)$, so we can find
another motorcycle that reaches $p$ at the same time in constant 
time using the second field associated with $p$ in $\mathcal D$.

Case (3a) corresponds to a positive answer to the test at 
Line \ref{line:check_2}. The condition $d(s_i, p') > d (s_i, \Top( S_i))$ 
detects whether the track of $i$ to the next target points hits any other
track. 
The other condition  $p' = \Top ( S_i)$ and $j\in M(p')$, 
checks for a boundary case, where $\Top(S_i)$
falls on another tentative track. The test is positive if 
$p'$ has already been identified as a target point of $j$ before. 
In this case we only extend the tentative track of $i$, without
doing any unnecessary halving. 

Line~\ref{line:check_3} branches to Case (3b). 
Similar to Line \ref{line:check_2}, we do not perform an
unnecessary halving operation when $i\in M(p')$ or $j \in M(p')$.

Our pseudocode does not handle explicitly the case where
two motorcycles have same supporting line. These cases can be easily
handled by ad-hoc arguments~\cite{ChengV07}. One way of doing it
is to insert additional target points at initialization.
For each supporting line shared by several motorcycles, between
any two consecutive motorcycles $i,j$ along this line that go 
toward each other, we
insert their potential collision point, that is, we insert into
$S_i$ and $S_j$ the point $p$ such that $\tau(i,p)=\tau(j,p)$.
For each motorcycle $i$ along this line, if 
the first starting point $s_j$ in the ray $(s_i,\vec v_i)$ 
is in $\overline{s_id_i}$, we also update $d_i$ to be $s_j$.

\begin{algorithm}
\caption{motorcycle\_graph}\label{algo1}
\begin{algorithmic}[1]
\State Initialize the dictionary $\mathcal D$. \Comment{Initialization}
\For {$i=1\to n$}
\State Set $c_i\gets s_i$, $ t_i \gets s_i$ and $S_i \gets \{ s_i, d_i \}$.
\State Insert $s_i$ and $d_i$ into $\mathcal D$; 
\State Insert motorcycle $i$ into $M(s_i)$ and $M(d_i)$. 
\EndFor
\State Initialize the event queue $Q$ and the data structures for
 	   ray-shooting and halving.
\While {${Q}$ is not empty} \Comment{Main loop}
    \State Let $(i,p)$ be the earliest available event. \Comment{So $p=t_i$.}
	\State Set $c_i\gets p$.
	\State Pop $\Top(S_i)$ from $S_i$. \Comment{Here $\Top(S_i)=p$.}
	\If {$\blocked(p)=\mbox{TRUE}$, or  $p=d_i$, or 
	\label{line:check_0}
	\State 	$\exists k\in M(p) \setminus\{i\}$ such that $\tau(k,p)=\tau(i,p)$}
	\label{line:check_1}
		\State Motorcycle $i$ crashes. 
		\State Set $t_i\gets p$. 
		\State Remove $i$ from $M(j)$ for all $j \in S_i$.
		\State Remove $S_i$ from $Q$.
	\Else 
    	\State $(j, p') \gets \operatorname{rayshooting} ( c_i, \vec{v}_i)$.
    	\If {$d(s_i, p') > d (s_i, \Top( S_i))$, 
			or ($p' = \Top ( S_i)$ and $j\in M(p')$)} \label{line:check_2}
			\State Set $t_i\gets \Top ( S_i)$.
		\Else  \label{line:check_3}
			\If {$i\in M(p')$}\label{line:check_4}
				\State Set $t_i\gets \Top ( S_i)$.
			\Else \label{line:check_5}
				\State Push $p'$ into $S_i$. 
				\State Push $p^*=h(c_i,p')$ into $S_i$.
				\State Set $t_i \gets p^*$.
				\State Insert $i$ into $M(p')$ and $M(p^*)$.
			\EndIf	
			\If	{$p'\notin\overline{s_j c_j}$ and $j\notin M(p')$}
			\label{line:check_6}	
				\State Push $p'$ into $S_j$.
				\State Push $p^*=h(c_j,p')$ into $S_j$.
				\State Set $t_j \gets p^*$.
				\State Insert $j$ into $M(p')$ and $M(p^*)$.
			\EndIf
		\EndIf
	\EndIf
	\State Set $\blocked(p) \gets \mbox{TRUE}$.
\EndWhile
\end{algorithmic}
\end{algorithm}

\pagebreak

\section{Example}\label{sec:example}
We give an example of the execution of our algorithm on a set of 4 motorcycles.
(Confirmed tracks are solid, and tentative tracks are dotted.)
It demonstrates two features of our algorithm, that were mentioned above.
\begin{itemize}
\item A tentative track may be longer than the final track in the motorcycle
graph. For instance, the tentative track $\overline{s_1d_1}$ in (b) is longer
than the final track $\overline{s_1\chi_{12}}$ in (q).
\item Our algorithm does not construct the motorcycle graph in chronological
order. For instance, in (i),  motorcycle 2 is moved to $\chi_{12}$, which
is its position at time $\tau(2,\chi_{12})=2.12072$.
Then in (k), motorcycle 3 is moved to $p_4$, which is its position
at time $\tau(3,p_4)=1.667206$.
\end{itemize}

The four motorcycles $1,2,3$ and $4$ start at time $0$ at initial points
$s_1=(0.8,3.3)$, $s_2=(0.5,1)$, $s_3=(5.7,0)$ and $s_4=(6, 3.4)$. 
Their velocities are $v_1 = 1.2(\cos -5^{\circ}, \sin -5^{\circ})$, 
$v_2 = 1.7(\cos 35^{\circ}, \sin 35^{\circ})$, 
$v_3 = 2(\cos 93^{\circ}, \sin 93^{\circ})$, and 
$v_4 = 0.8(\cos -37^{\circ}, \sin -37^{\circ})$. 

We use the halving scheme as specified in Section~\ref{sec:notation} with $\rho=1/2$.
So for instance, we create $p_4$ in (j) by halving $\overline{s_3\chi_{23}}$.
There are three crossings along this segment: $\chi_{13},\chi_{23},\chi_{34}$.
Then $p_4$ is created as a point between $\chi_{13}$ and $\chi_{34}$, in this
case we just use the midpoint. 
\begin{figure}[hbt!]
	\centering
	\begin{subfigure}[b]{0.45\textwidth}
		\centering\includegraphics[width=\textwidth]{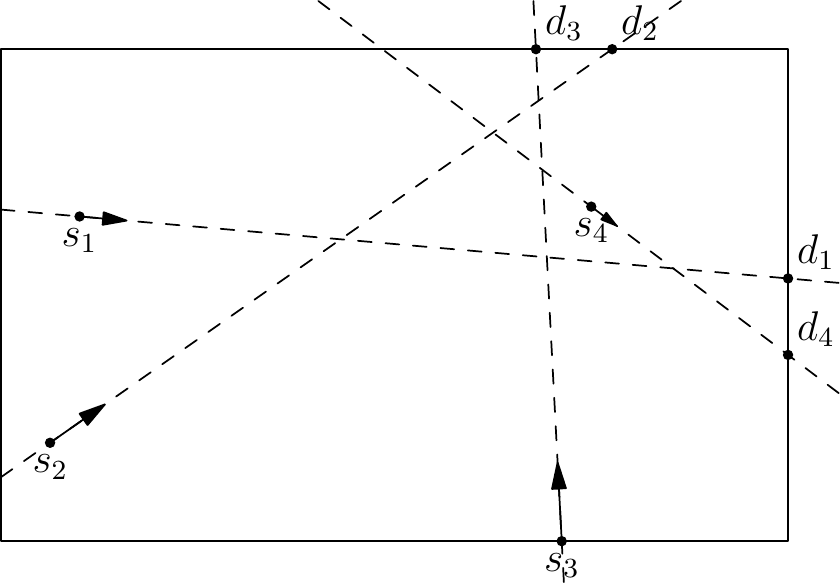}
		\caption{}\label{eg:0}
	\end{subfigure}
	\hspace{2ex}
	\begin{subfigure}[b]{0.45\textwidth}
		\centering\includegraphics[width=\textwidth]{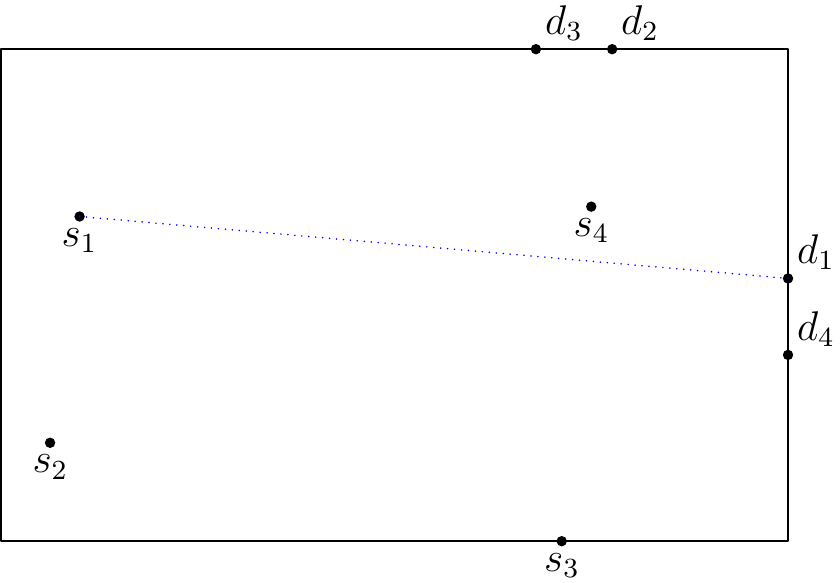}
		\caption{$\tau(1,d_1)=6.022919$}\label{eg:1}
	\end{subfigure}
	\begin{subfigure}[b]{0.45\textwidth}
		\centering\includegraphics[width=\textwidth]{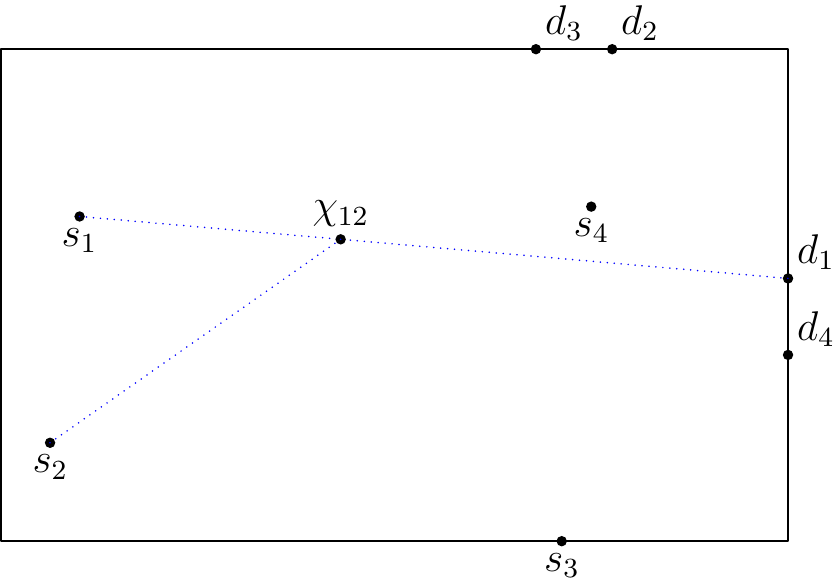}
		\caption{new events: $\tau(1, \chi_{12})=2.219469$, 
		and $\tau(2,\chi_{12})=2.120721$}\label{eg:2}
	\end{subfigure}
	\hspace{2ex}
	\begin{subfigure}[b]{0.45\textwidth}
		\centering\includegraphics[width=\textwidth]{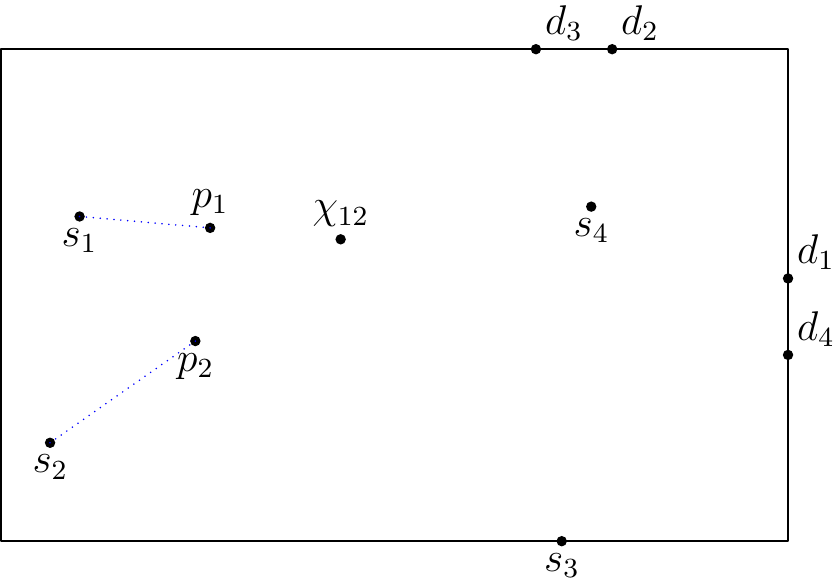}
		\caption{new events: $\tau(1,p_1)=1.109735$, and $\tau(2,p_2)=1.060361$}\label{eg:3}
	\end{subfigure}
\end{figure}
\begin{figure}[hbt!]
	\begin{subfigure}[b]{0.45\textwidth}
		\centering\includegraphics[width=\textwidth]{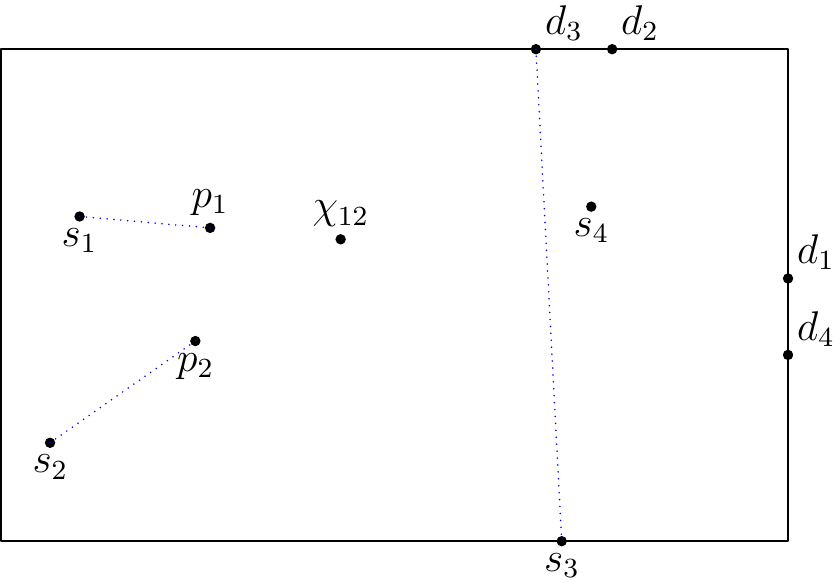}
		\caption{$\tau(3,d_3)=2.503431$}\label{eg:4}
	\end{subfigure}
	\hspace{2ex}
	\begin{subfigure}[b]{0.45\textwidth}
		\centering\includegraphics[width=\textwidth]{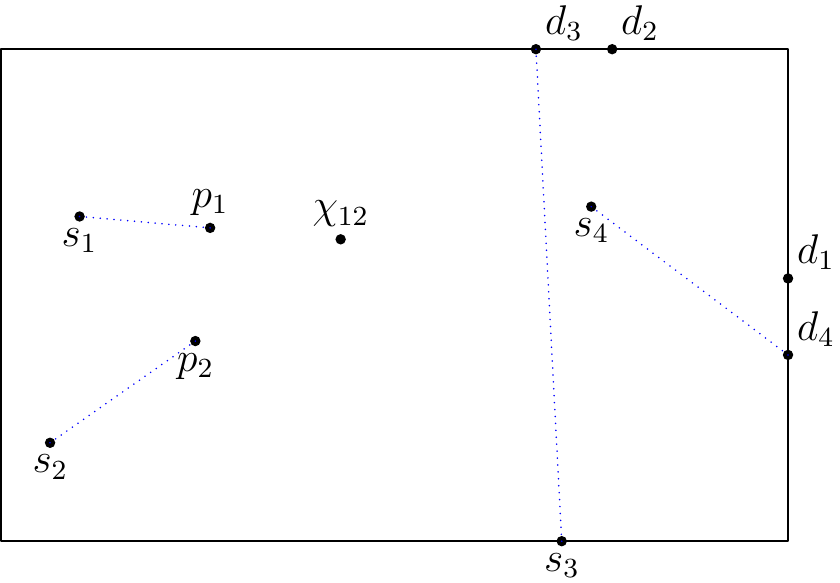}
		\caption{$\tau(4,d_4)=3.130339$}\label{eg:5}
	\end{subfigure}
	\centering
	\begin{subfigure}[b]{0.45\textwidth}
		\centering\includegraphics[width=\textwidth]{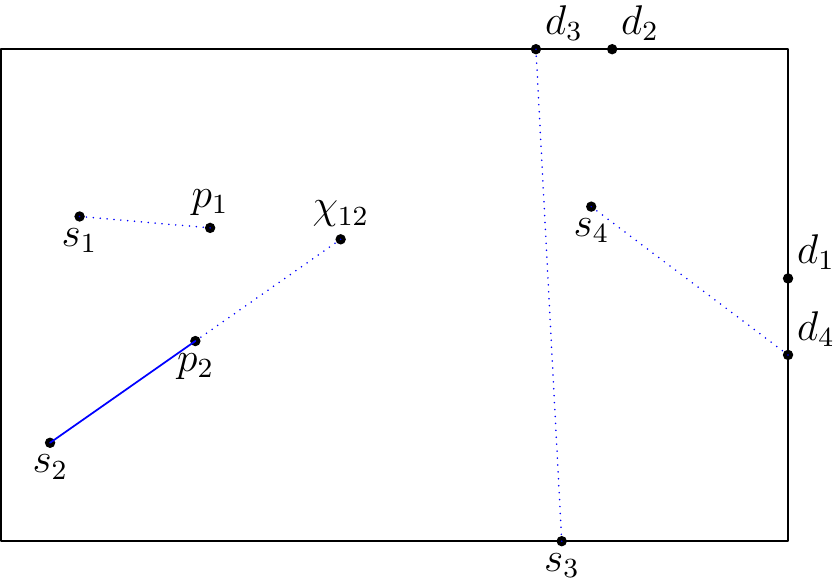}
		\caption{move: $\tau(2,p_2)=1.060361$}\label{eg:6}
	\end{subfigure}	
	\hspace{2ex}
	\begin{subfigure}[b]{0.45\textwidth}
		\centering\includegraphics[width=\textwidth]{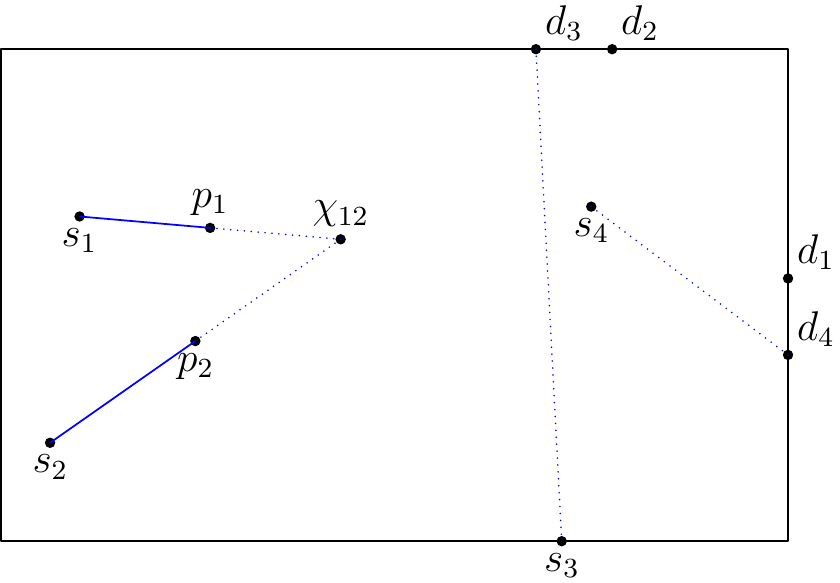}
		\caption{move: $\tau(1,p_1)=1.109735$}\label{eg:7}
	\end{subfigure}
	\begin{subfigure}[b]{0.45\textwidth}
		\centering\includegraphics[width=\textwidth]{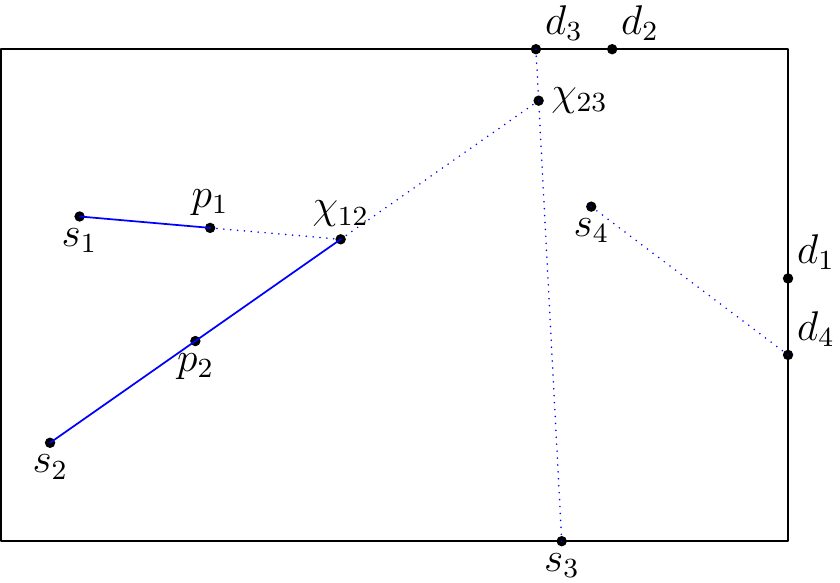}
		\caption{move: $\tau(2,\chi_{12})=2.120721$. new events:
		$\tau(2,\chi_{23})=3.565653$, and $\tau(3,\chi_{23})=2.24147$
		}\label{eg:8}
	\end{subfigure}
	\hspace{2ex}
	\begin{subfigure}[b]{0.45\textwidth}
		\centering\includegraphics[width=\textwidth]{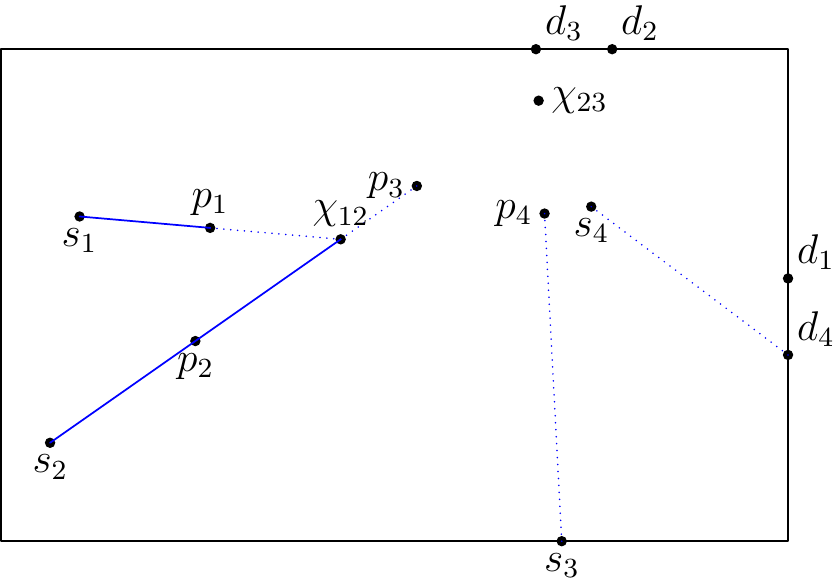}
		\caption{new events: $\tau(2,p_3)=2.676739$, 
		and $\tau(3,p_4)=1.667206$ }\label{eg:9}
	\end{subfigure}
\end{figure}
\pagebreak
\begin{figure}[htb!]
	\begin{subfigure}[b]{0.45\textwidth}
		\centering\includegraphics[width=\textwidth]{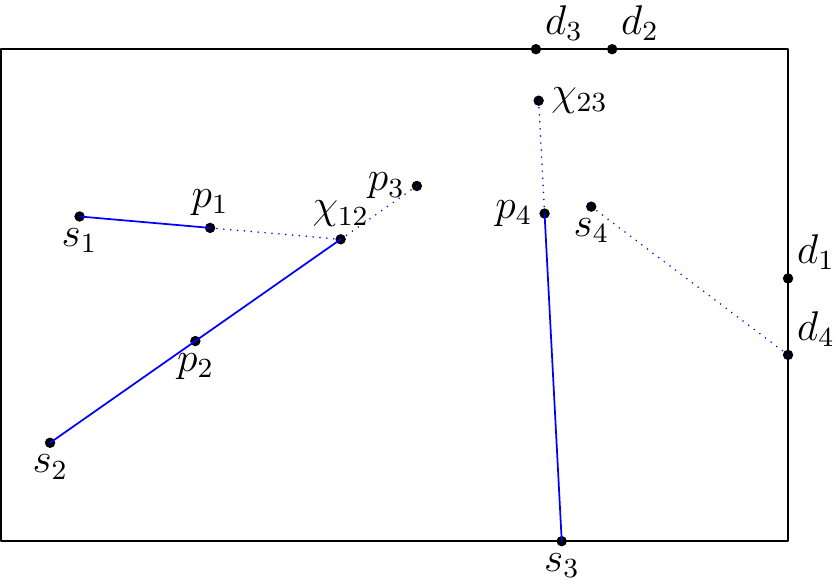}
		\caption{move: $\tau(3,p_4)=1.667206$}\label{eg:10}
	\end{subfigure}
	\hspace{2ex}
	\begin{subfigure}[b]{0.45\textwidth}
		\centering\includegraphics[width=\textwidth]{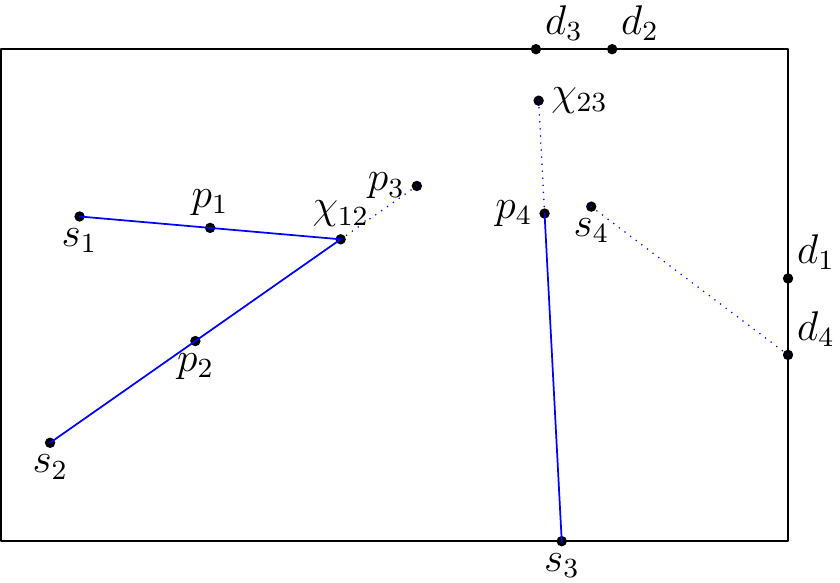}
		\caption{move: $\tau(1,\chi_{12})=2.219469$}\label{eg:11}
	\end{subfigure}
	\begin{subfigure}[b]{0.45\textwidth}
		\centering\includegraphics[width=\textwidth]{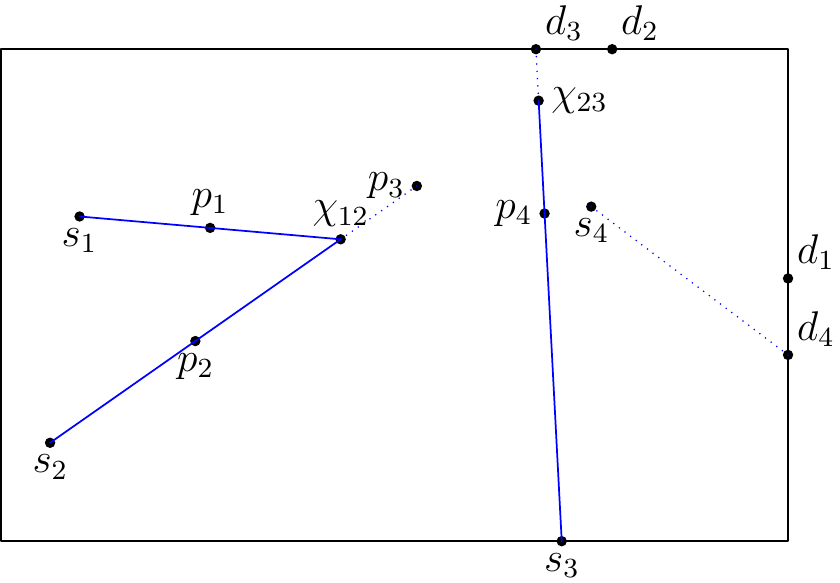}
		\caption{move: $\tau(3,\chi_{23})=2.24147$}\label{eg:12}
	\end{subfigure}
	\hspace{2ex}
	\begin{subfigure}[b]{0.45\textwidth}
		\centering\includegraphics[width=\textwidth]{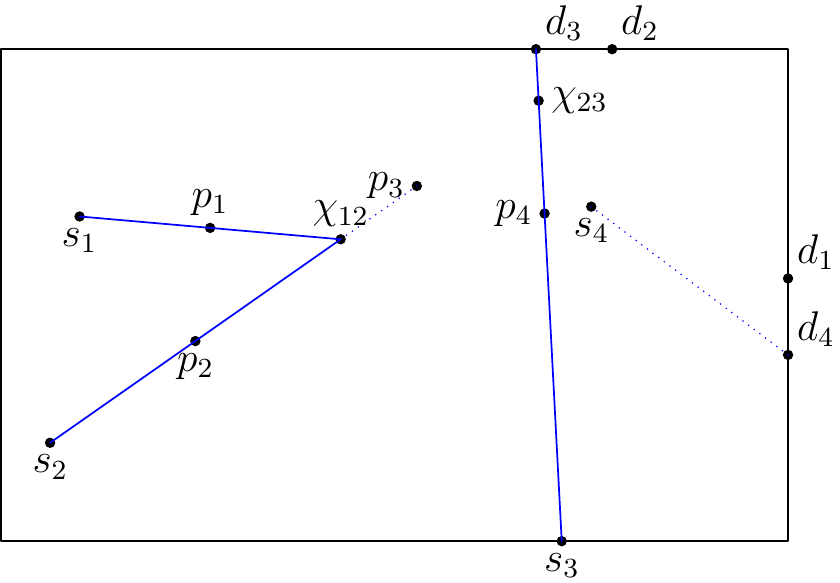}
		\caption{move: $\tau(3,d_3)=2.503431$}\label{eg:13}
	\end{subfigure}
	\begin{subfigure}[b]{0.45\textwidth}
		\centering\includegraphics[width=\textwidth]{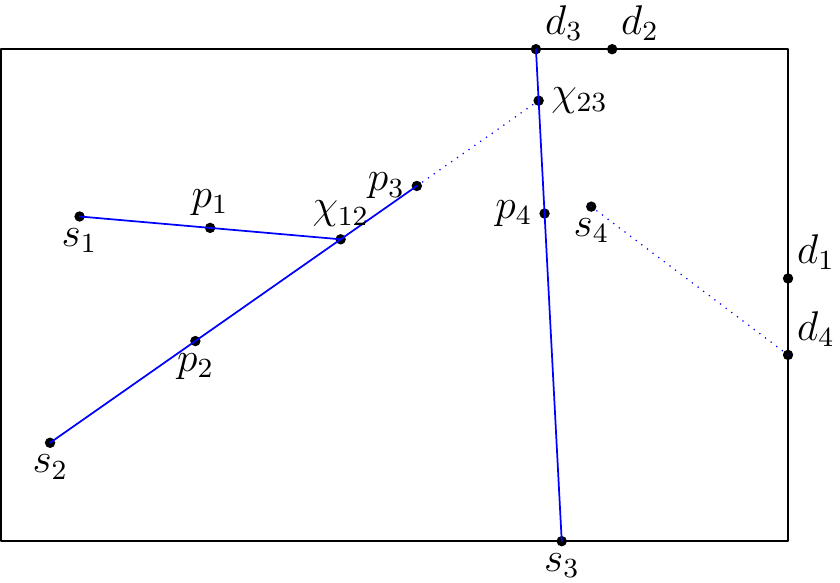}
		\caption{move: $\tau(2,p_3)=2.676739$}\label{eg:14}
	\end{subfigure}
	\hspace{2ex}
	\begin{subfigure}[b]{0.45\textwidth}
		\centering\includegraphics[width=\textwidth]{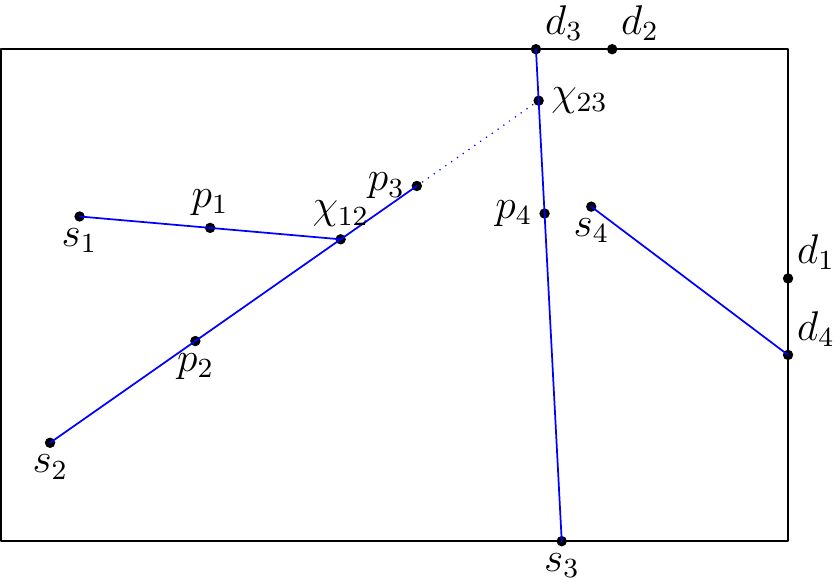}
		\caption{move: $\tau(4,d_4)=3.130339$}\label{eg:15}
	\end{subfigure}
	\begin{subfigure}[b]{0.45\textwidth}
		\centering\includegraphics[width=\textwidth]{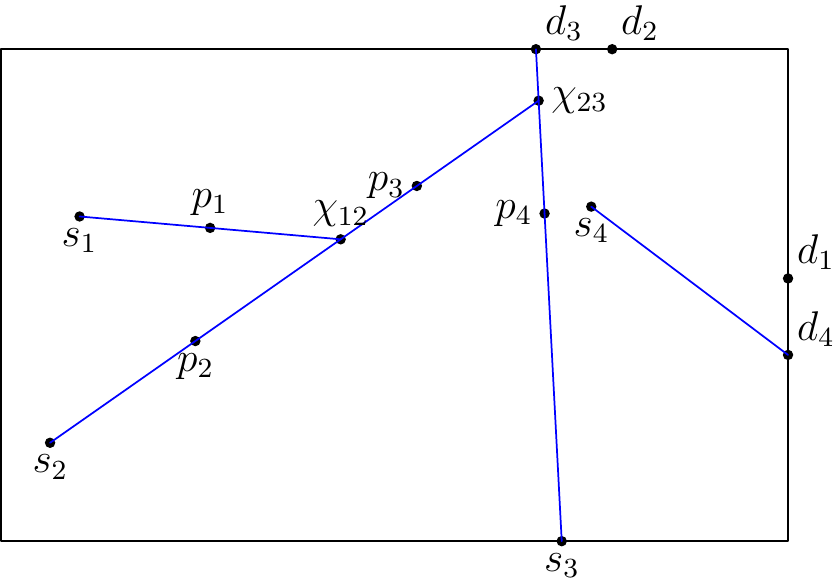}
		\caption{move: $\tau(2,\chi_{23})=3.565653$}\label{eg:16}
	\end{subfigure}
\end{figure}

\end{document}